\documentclass[english,a4paper,10pt]{article}
\usepackage[utf8]{inputenc}
\usepackage[english]{babel}

\usepackage{framed}
\usepackage{amssymb}
\usepackage{amsmath}
\usepackage{todonotes}
\usepackage{amsthm}
\usepackage{scalerel}
\usepackage{hyperref}
\usepackage[vlined,ruled,noend,linesnumbered]{algorithm2e}
\SetArgSty{text}
\DontPrintSemicolon
\usepackage{bussproofs}
\EnableBpAbbreviations
\usepackage{mathtools}
\usepackage{tikz}
\usetikzlibrary{shapes.geometric, patterns}
\usetikzlibrary{arrows.meta}
\usetikzlibrary{calc}
\usetikzlibrary{shapes.misc, positioning}
\usepackage{csquotes}
\usepackage{booktabs}
\usepackage{thmtools}
\usepackage{thm-restate}
\usepackage{enumitem}
\usepackage{pdfpages}
\usepackage{nameref}

\newif\iftr
\trtrue

\usepackage{stmaryrd}

\declaretheorem[name=Theorem]{theorem}
\newtheorem{definition}{Definition}
\newtheorem{lemma}[theorem]{Lemma}
\newtheorem{example}{Example}
\newtheorem*{remark}{Remark}

\newcommand{\added}[1]{#1}

\newcommand{\ttodo}[4]{\ifthenelse{\equal{#1}{inline}}{\todo[inline, author=#2, color = 
#3]{#4}}{\todo[color=#3]{#2: #4}}}

\newcommand{\hide}[1]{}

\newcommand{\Wlog}{W.l.o.g.\ }
\newcommand{\wrt}{w.r.t.\ }
\newcommand{\st}{s.t.\ }
\newcommand{\ie}{i.e.\ }
\newcommand{\eg}{e.g.\ }

\newcommand{\cf}{cf.\ }
\newcommand{\vs}{vs.\ }

\newlength{\myl}
\newcommand{\longsquigarrow}[1]{
    \settowidth{\myl}{$~_{#1}$}
    \raisebox{-0.01cm}{\xymatrix@C=\myl{
            {}\ar@{~>}[r]^{~_{#1}}&{}
        }
    }
}

\newcommand{\ACzero}{\ensuremath{\mathsf{AC^0}}\xspace}

\newcommand{\DL}{\textsl{DL-Lite}}

\newcommand{\OP}{\ensuremath{\mathsf{OP}}\xspace}
\newcommand{\sk}{\ensuremath{_\mathsf{sk}}\xspace}
\newcommand{\cq}{\ensuremath{_\mathsf{cq}}\xspace}
\newcommand{\leaf}{\ensuremath{\mathsf{leaf}}\xspace}
\newcommand{\sink}{\ensuremath{\mathsf{sink}}\xspace}

\newcommand{\card}[1]{\lvert #1\rvert}

\newcommand{\PTime}{\text{\upshape{\textsc{P}}}\xspace}
\newcommand{\PSpace}{\text{\upshape{\textsc{PSpace}}}\xspace}
\newcommand{\NP}{\text{\upshape{\textsc{NP}}}\xspace}

\newcommand{\ExpTime}{\text{\upshape{\textsc{ExpTime}}}\xspace}

\newcommand{\NExpTime}{\text{\upshape{\textsc{NExpTime}}}\xspace}

\newcommand{\LogSpace}{\text{\upshape{\textsc{LogSpace}}}\xspace}
\newcommand{\Amc}{\ensuremath{\mathcal{A}}\xspace}

\newcommand{\Kmc}{\ensuremath{\mathcal{K}}\xspace}
\newcommand{\Lmc}{\ensuremath{\mathcal{L}}\xspace}

\newcommand{\Tmc}{\ensuremath{\mathcal{T}}\xspace}

\newcommand{\ALC}{\ensuremath{\mathcal{ALC}}\xspace}

\newcommand{\EL}{\ensuremath{\mathcal{E}\hspace{-0.1em}\mathcal{L}}\xspace}

\newcommand{\ALCHOI}{\ensuremath{\mathcal{ALCHOI}}\xspace}
\newcommand{\ALCHOIQ}{\ensuremath{\mathcal{ALCHOIQ}}\xspace}

\newcommand{\Horn}[1]{\textsl{Horn-}#1}
\newcommand{\HornALC}{\Horn{\ALC}}

\newcommand{\HornALCHOI}{\Horn{\ALCHOI}}
\newcommand{\HornALCHOIQ}{\Horn{\ALCHOIQ}}

\renewcommand{\L}{\ensuremath{\mathcal{L}}\xspace}
\newcommand{\Lcq}{\ensuremath{\L_{\textit{cq}}}\xspace}
\newcommand{\DLLite}{\ensuremath{\textsl{DL-Lite}}\xspace}
\newcommand{\DLLiteR}{\ensuremath{\textsl{DL-Lite}_R}\xspace}

\newcommand{\q}{\ensuremath{\mathbf{q}}\xspace} 
\newcommand{\x}{\ensuremath{\vec{x}}\xspace}
\newcommand{\y}{\ensuremath{\vec{y}}\xspace}
\renewcommand{\a}{\ensuremath{\vec{a}}\xspace}

\newcommand{\tup}[1]{\langle #1 \rangle}

\newcommand{\NC}{\ensuremath{\textsf{N}_\textsf{C}}\xspace}
\newcommand{\NR}{\ensuremath{\textsf{N}_\textsf{R}}\xspace}

\newcommand{\NI}{\ensuremath{\textsf{N}_\textsf{I}}\xspace}

\newcommand{\sig}{\ensuremath{\textsf{\upshape{sig}}}\xspace}

\newcommand{\answer}{\ensuremath{\vec{a}}\xspace}

\newcommand{\ex}[1]{\ensuremath{\textsf{\upshape{#1}}}\xspace}
 
\newcommand{\p}{\ensuremath{\mathcal{P}}\xspace}

\newcommand{\R}{\ensuremath{\mathfrak{D}}\xspace}

\newcommand{\Rcq}{\ensuremath{\R_\textsl{cq}}\xspace}
\newcommand{\Rsk}{\ensuremath{\R_\textsl{sk}}\xspace}

\newcommand{\el}{\ensuremath{\ell}\xspace}

\newcommand{\ds}{\ensuremath{\mathcal{D}}\xspace}

\newcommand{\MOPO}{\upshape{(\ensuremath{\mathbf{MP_e}})}\xspace}
\newcommand{\CONJ}{\upshape{(\ensuremath{\mathbf{C_e}})}\xspace}
\newcommand{\EQUAL}{\upshape{(\ensuremath{\mathbf{E_e}})}\xspace}
\newcommand{\TAUT}{\upshape{(\ensuremath{\mathbf{T_e}})}\xspace}
\newcommand{\EXISTS}{\upshape{(\ensuremath{\mathbf{G_e}})}\xspace}
\newcommand{\gMOPO}{\upshape{(\ensuremath{\mathbf{MP}})}\xspace}
\newcommand{\gCONJ}{\upshape{(\ensuremath{\mathbf{C}})}\xspace}
\newcommand{\gEQUAL}{\upshape{(\ensuremath{\mathbf{E}})}\xspace}
\newcommand{\gEXISTS}{\upshape{(\ensuremath{\mathbf{G}})}\xspace}

\newcommand{\m}{\ensuremath{\mathfrak{m}}\xspace}
\newcommand{\mx}{\ensuremath{\mathfrak{m}_{\mathsf{x}}}\xspace}
\newcommand{\mtree}{\ensuremath{\m_{\mathsf{t}}}\xspace}
\newcommand{\msize}{\ensuremath{\m_{\mathsf{s}}}\xspace}
\newcommand{\mdomain}{\ensuremath{\m_{\mathsf{d}}}\xspace}

\newcommand{\dom}{\textsf{dom}\xspace}

\newcommand{\In}{\text{\upshape{\textsc{In}}}\xspace}
\newcommand{\Out}{\text{\upshape{\textsc{Out}}}\xspace}
\newcommand{\Size}{\text{\upshape{\textsc{Size}}}\xspace}
\newcommand{\Inds}[1]{\textsf{ind}(#1)\xspace}

\begin{document}
\title{Explaining Ontology-Mediated Query Answers using Proofs over Universal
Models \\(Technical Report)}
%
%
\author{Christian Alrabbaa
\and
Stefan Borgwardt
\and
Patrick Koopmann
\and
Alisa Kovtunova
}
%
%
%
\date{}
\maketitle

\begin{abstract}
In ontology-mediated query answering, access to incomplete data sources is
mediated by a conceptual layer constituted by an ontology, which can be
formulated in a description logic (DL) or using existential rules.
In the literature, there exists a
multitude of complex techniques for incorporating ontological knowledge into
queries.
However, few of these approaches were designed for explainability of the
query answers. We tackle this challenge by adapting
an existing proof framework toward conjunctive query
answering, based on the notion of universal models.
We investigate the data and combined complexity of determining the
existence of a proof below a given quality threshold, which can be measured in
different ways. By distinguishing various parameters such as the shape of the
query, we obtain an overview of the complexity of this problem for several Horn
DLs.
\end{abstract}

\section{Introduction}

Description logics (DLs) are a family of knowledge representation formalisms that can be seen as decidable fragments of first-order logic using only unary and binary predicates~\cite{BHLS-17}.
This family contains very expressive DLs like $\mathcal{SROIQ}$, which underlies the standardized Web Ontology Language OWL~2,\footnote{\url{https://www.w3.org/TR/owl2-overview/}} as well as the light-weight DLs \DLLiteR and \EL, corresponding to the OWL~2 profiles QL and EL, respectively.
%
We focus here on \emph{Horn DLs} up to \Horn\ALCHOI~\cite{DBLP:conf/ijcai/OrtizRS11,DBLP:journals/tocl/KrotzschRH13}, whose axioms can be expressed as \emph{existential rules} (with equality)~\cite{DBLP:journals/ws/CaliGL12}.
%
%
The complexity of standard reasoning problems such as entailment of axioms or
facts (ground atoms) from an ontology (a finite set of axioms) has been studied for decades and
is well-understood by now~\cite{DBLP:conf/kr/OrtizRS10,BHLS-17}.
%
Another popular reasoning problem for DLs is that of \emph{ontology-mediated query answering (OMQA)}, which generalizes query answering over databases by allowing to query implicit knowledge that is inferred by the ontology~\cite{JAR-2007,DBLP:conf/ijcai/OrtizRS11}.
%
%

\emph{Explaining} DL reasoning has a long tradition, starting with the first works on \emph{proofs} for standard DL entailments~\cite{DeMc-96,DBLP:conf/ecai/BorgidaFH00}.
A popular and very effective method is to compute \emph{justifications}, which simply point out the axioms 
from the ontology that are responsible for an 
entailment~\cite{ScCo03,DBLP:conf/ki/BaaderPS07,Pena-09,Horr-11}.
More recently, work has resumed on techniques to find proofs for explaining more complex logical consequences~\cite{DBLP:conf/semweb/HorridgePS10,KaKS-DL17,LPAR23:Finding_Small_Proofs_for,ABB+-DL20,ABB+-CADE21}.
On the other hand, if a desired entailment does \emph{not} hold, one needs different explanation 
techniques such as
abduction~\cite{DBLP:conf/ijcai/Koopmann21,EX_RULES_ABDUCTION,%
DBLP:conf/kr/CalvaneseOSS12} or
counterinterpretations~\cite{DBLP:conf/ki/AlrabbaaHT21}. 
Explaining answers to conjunctive queries (CQs) has also been investigated before, in the form of abduction for missing 
answers over \DLLite ontologies~\cite{DBLP:conf/kr/CalvaneseOSS12}, provenance for positive answers in \DLLite and 
\EL~\cite{DBLP:conf/ijcai/CalvaneseLOP019,DBLP:conf/ijcai/BourgauxOPP20}, as proofs for \DLLite query 
answering~\cite{DBLP:conf/otm/BorgidaCR08,Stefanoni-11,DBLP:conf/ekaw/CroceL18}, as well as proofs and provenance for rule reasoning~\cite{ElKM-RuleMLRR22,DBLP:conf/grades/RamusatMS22}.
Inspired by the latter,
we also investigate proofs for CQ answers, but consider more expressive DLs and want to 
find \emph{good} proofs according to different quality measures.
We focus on Horn DLs, for which every ontology has a \emph{universal model} that captures
exactly the query answers over the ontology~\cite{DBLP:journals/ws/CaliGL12}. 
While classically models are used for explaining missing entailments~\cite{DBLP:conf/ki/AlrabbaaHT21}, 
this property allows us to use universal models also to explain positive query answers. 
%
%

\begin{figure}[tb]
\centering\begin{footnotesize}
  \begin{tikzpicture}[scale=0.8,
  block/.style={
  draw, minimum width=20pt, inner sep = 5pt,
  rounded rectangle,
  },
  tbox/.style={
  draw=gray, text=gray,
  circle,
  },
  he/.style={rounded corners=10pt,->},
  match/.style={green!50!black}]
  \path  (0,0) node[block] (x0) {$\ex{a}\smash{\textcolor{gray}{{}=x'}}$}
   ([shift={(0:1.2)}]x0) node[rectangle,draw=blue] (y0) {$A$}
   ([shift={(0:-1.1)}]x0) node[match] (y5) {$D$}
     (-1,1.8) node[block] (x1) {{\vphantom{$a$}\smash{$f(\ex{a})\textcolor{gray}{{}=y}$}}}
      ([shift={(0:1.3)}]x1) node (y2) {$B$}
      ([shift={(0:-1.3)}]x1) node (y4) {$E$}
     (-2,3.8) node[block] (x2) {{\vphantom{$a$}\smash{$g(f(\ex{a}))$}}};
       \draw[he] (x0)--  ($(x0)!0.5!(x1)+(-0.5,-0.3)$) node[match,above left=-2pt] (y1) {$r$} --(x1);
       \draw[he] (x1)-- ($(x2)!0.5!(x1)+(-0.5,-0.3)$) node[above left=-2pt] (y3) {$s$} --(x2);
       \draw[he,dashed, gray] (y0)--  ($(y0|-y1)+(-0.3,0)$) -- ($(y1)+(0.4,0)$);
       \draw[he,dashed, gray] (y0)|-(y2);
       \draw[he,dashed, gray] (y2)--  ($(y2|-y3)+(-0.3,0)$) -- ($(y3)+(0.4,0)$);
        \draw[he,dashed, gray] (y1)-- ($(y1)-(1.7,0)$) -- ($(y4)-(1,0)$) -- (y4);
        \draw[he,dashed, gray] (y3) -- ($(y3)-(0.8,0)$) --  ($(y4)-(1,0)$) -- (y4);
        \draw[he,dashed, gray] (y4)-- (y4|-y5) -- (y5);
        \draw[he,dashed, gray] (y1) -- ($(y1)-(.9,0)$) --  ($(y5)-(.9,0)$) -- (y5);
        \path  (-6.2,0) node[block, fill=gray!20] (z3) {$\ex{a}$}
        (-4.7,0) node[block,fill=gray!20] (z0) {\vphantom{$a$}\smash{$x'$}}
   ([shift={(0:0.6)}]z0) node[match] (z1) {$D$}
     (-5.45,1.8) node[block,fill=gray!20] (z2) {\vphantom{$a$}\smash{$y$}};
     \draw[he] (z0) edge node[match,right] {$r$} (z2);
     \draw[he] (z3) edge node[match,left] {$r$} (z2);
    \end{tikzpicture}
  \end{footnotesize}
\caption{\added{The query (on the left) and the relevant part of the universal model (on the right) 
from 
Example~\ref{ex}.}}
 \label{fig:uni-model}
\end{figure}
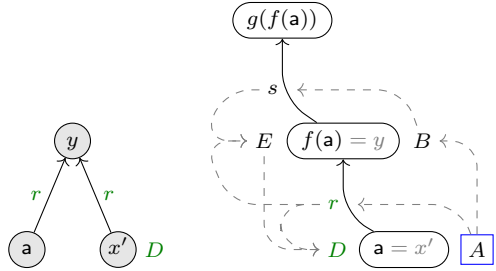
%
 \begin{example}\label{ex}
  Consider the fact $A(\ex{a})$, the existential rules (which can be expressed in \Horn\ALCHOI)
  %
  \begin{align*}
  A(x)&\to\exists y.\,r(x,y) \land B(y), &
  s(x,y) \land  r(z,x) &\to E(x), \\
  B(x)&\to\exists z.\,s(x,z) \land A(z), &
  E(x) \land r(y,x) &\to D(y),&
  \end{align*}
  and the \emph{conjunctive query} $\q(x) = \exists x',y.\, r (x, y) \,\land\,  r (x', y) \,\land\, D(x').$ 
  Individual $\ex{a}$ is an \emph{answer} to $\q$ in this ontology. The query instantiated with this answer 
  is depicted on the left in Fig.~\ref{fig:uni-model}, using edges for binary predicates
  and node labels for unary predicates. 
%
  To explain the answer, we show on the right of the figure the relevant part of the universal model of the
ontology, where unary and binary predicates are represented similarly. The nodes represent objects in the 
model and are identified by \emph{Skolem terms}, together with the assignments to the variables in the 
query. 
For example, 
$f(\ex{a})$ can be described as \enquote{the $r$-successor of~$\ex{a}$}, which has to be present 
in any model of the ontology due to the first rule. 
\added{The Skolem functions like~$f$ and~$g$ are created uniquely for each existentially quantified variable in the rules}.
In addition to explaining how the query is matched to the universal model, the 
dashed gray edges 
%
  indicate a \emph{proof} of $\q(\ex{a})$. For instance, $A(\ex{a})$, together with the first rule, implies 
  the existence of the $r$-successor satisfying $B$,
  and $D(\ex{a})$ follows from $E(f(\ex{a}))$ and $r(\ex{a},f(\ex{a}))$ through the last rule. 
  To make this representation more accessible for larger proofs, in real applications we would show proof 
  steps only on 
  demand, whenever a user selects a fact to be explained in the model.

\end{example}

In previous work~\cite{LPAR23:Finding_Small_Proofs_for,ABB+-CADE21}, we developed a 
\added{ formal } framework for proofs 
\added{ in } standard DL reasoning.
We 
investigated the complexity of finding \emph{small} proofs according to different proof measures: 
\emph{(proof) size}, \ie the number of distinct formulas in a proof, and \emph{(proof) tree size}, 
\ie the size when the proof is presented as tree, as it is done often in 
practice~\cite{KaKS-DL17,EvonneIJCAR22}. 
In this framework, proofs are generated by a so-called 
\emph{deriver} that specifies which inferences are possible in a proof.

%
\begin{table}[t]
\centering%
\caption{Summary of the combined complexity results for $\OP\sk(\L,\m)$.}%
\label{table:summary}%
\tabcolsep=6pt%
\begin{tabular}{c c c c c c}
  \toprule
 &  \multicolumn{2}{c|}{\DL} &   \multicolumn{2}{c|}{\EL} & {\HornALCHOI} \\
Measure & tree-sh. & CQ & IQ  & CQ & CQ  \\
  \midrule
 \raisebox{6pt}{Domain size} & \multicolumn{2}{c|}{\shortstack{\NP-c \\
{\scriptsize [Th.\,\ref{thm:dllite-el-upper-bounds},\ref{the:np-hard-trees}]}}}
& \multicolumn{2}{c|}{\shortstack{in \ExpTime \\ {\scriptsize
[Th.\,\ref{thm:dllite-el-upper-bounds}]}}} & {\shortstack{in
\NExpTime \\ {\scriptsize [Th.\,\ref{th:HornSROIQ}]}}} \\
 \raisebox{6pt}{Tree size}  & \shortstack{in \PTime \\ {\scriptsize
[Th.\,\ref{th:ptime-combined}]}}  & \shortstack{\NP-c \\ {\scriptsize
[Th.\,\ref{thm:dllite-el-upper-bounds},\ref{thm:dllite-lower}]}} &\shortstack{\PTime-c \\ {\scriptsize
[Th.\,\ref{the:el-bounds}]}} & \shortstack{\NP-c \\ {\scriptsize
[Th.\,\ref{the:el-bounds}]}} & 
 \shortstack{ in \PSpace \\ {\scriptsize [Th.\,\ref{th:HornSROIQ}]}}
 \\
 \raisebox{6pt}{Proof size} & \multicolumn{2}{c|}{\shortstack{\NP-c \\
{\scriptsize [Th.\,\ref{thm:dllite-el-upper-bounds},\ref{the:np-hard-trees}]}}}
& \multicolumn{3}{c}{\shortstack{in \ExpTime \\ {\scriptsize
[Th.\,\ref{th:HornSROIQ-IQ}]}}} 
\\
 \midrule
 \shortstack{Proof size \\ bound}
 & \multicolumn{2}{c|}{\shortstack{polynomial \\ {\scriptsize [Lem.\,\ref{lem:size-lower-light}]} }} & 
 \multicolumn{2}{c|}{\shortstack{exponential \\ {\scriptsize [Lem.\,\ref{lem:size-lower-light}]} }}  & 
 \shortstack{double exponential \\ {\scriptsize [Lem.\,\ref{lem:size-lower}]}
}
\\
 \bottomrule
\end{tabular}
\end{table}

\added{To be able to reuse results, the 
present work develops proofs for query answers within the same framework.}
In particular, in order to explain query answers using universal models, we introduce a special deriver that 
applies to a large 
family of Horn-DLs, and in which inferences in the proof
directly correspond to the construction of the universal model.
For such proofs, 
if we visualize them as in the example, another proof measure becomes relevant:
the \emph{domain size}, which is the number of elements from the universal model that are used in the proof. 
In the example, the domain size of the proof is 3. 
After introducing our deriver, we investigate the complexity of finding good proofs w.r.t. the different 
measures, as 
well as bounds on the size of the obtained proofs. An overview of our results is shown in 
Table~\ref{table:summary}. Because it 
introduces fresh objects, our deriver is \added{only sound for a Skolemized version of the TBox, and not 
for the original TBox.}
At the end of the paper, we have 
a brief 
look at another deriver in which all inferences are sound \wrt the original TBox, and argue that, 
while the complexity of the resulting decision problem is often similar,
this 
deriver is less helpful in explaining query answers to users. This paper extends initial results in this 
direction 
from a workshop paper~\cite{DL22paper}.

Proof details can be found in the appendix.

\section{Preliminaries}\label{sec:preliminaries}



\paragraph{Logics.}
We assume basic knowledge about first-order logic and familiarity with terminology such as variables, terms, atoms, sentences, etc.
%
%
Throughout the paper, we use \L to refer to fragments of first-order logic.
DLs are fragments of the two-variable fragment, for which we assume 
unary predicates to be taken from a countably infinite set $\NC$ of \emph{concept names}, 
binary predicates to be taken from a countably infinite set $\NR$ of \emph{role names},
and constants to be taken from a countably infinite set $\NI$ of \emph{individual names}~\cite{BHLS-17}.
Moreover, we use $\top$ and $\bot$ as special concept names that are always satisfied or 
\added{always} not satisfied, respectively.
We focus on \emph{Horn DLs} that can be represented using existential rules with 
equality~\cite{DBLP:journals/ws/CaliGL12}. An existential rule is 
a first-order sentence of the form
$\forall\y,\vec{z}.\,\psi(\y,\vec{z})\to\exists\vec{u}.\,\chi(\vec{z},\vec{u})$, with the
\emph{body} $\psi(\y,\vec{z})$ and the \emph{head} $\chi(\vec{z},\vec{u})$ being conjunctions of atoms of the form $A(t_1)$, $R(t_1,t_2)$, or $t_1=t_2$, 
where $t_1$ and $t_2$ are constants or variables from 
$\vec{z}$, $\vec{u}$ and~$\y$. 
We usually omit the universal quantification.

\begin{table}[tb]
  \centering
  \caption{Sentences of $\HornALCHOI$, where 
  $A,B,C\in\NC$,
  $a\in\NI$,
   $R,R_1,R_2$ are \emph{roles} of the form~$r$ or~$r^-$ (\emph{inverse role}), $r\in\NR$, 
  and we identify $r^-(x,y)$ with $r(y,x)$.
  }
  \label{tab:normal-form}
  \begin{tabular}{r@{\quad}l@{\quad}l}
    \toprule
    (i) & $A\sqsubseteq B$ & $A(x)\to B(x)$\\
    (ii) & $A\sqcap B\sqsubseteq C$ & $A(x)\land B(x)\to C(x)$ \\
    (iii) & $\exists R.A\sqsubseteq B$ & $R(x,y)\land A(y)\to B(x)$ \\
    (iv) & $A\sqsubseteq\exists R.B$ & $A(x)\to\exists y.\,R(x,y)\land B(y)$ \\
    (v) & $A\sqsubseteq\forall R.B$ & $A(x)\land R(x,y) \to B(y)$ \\
    (vi) & $A\sqsubseteq\{a\}$ & $A(x)\to x=a$ \\
    (vii) & $R_1\sqsubseteq R_2$ & $R_1(x,y) \to R_2(x,y)$ \\
    \bottomrule
  \end{tabular}
\end{table}
For DLs, one usually uses a different, dedicated syntax. Table~\ref{tab:normal-form} shows the allowed rules in 
$\HornALCHOI$, together with their representation in DL syntax, where, for simplicity, we assume the rules to be 
\emph{normalized}. A set $\Tmc$ of such rules is called \emph{TBox} or \emph{ontology}.
In $\HornALC$, only expressions of the forms (i)--(v) without inverse roles are allowed, 
\EL further restricts \Horn\ALC by disallowing~(v) and~$\bot$, and 
\DLLite only allows expressions $R_1\sqsubseteq R_2$, $A\sqsubseteq C$, $C\sqsubseteq A$, and $A\sqcap B\sqsubseteq\bot$, 
where $R_1$, $R_2$ are (possibly inverse) roles, $A,B\in\NC$, and $C$ is either a concept name or $\exists R.\top$, for a (possibly inverse) role~$R$.

\paragraph{Query Answering.} An \emph{ABox}~\Amc is a set of ground atoms (called \emph{facts}) of the form
$A(a)$ or $r(a,b)$, which together with a TBox \Tmc forms a \emph{knowledge base} (KB) 
$\Kmc=\Tmc\cup\Amc$. 
Its \emph{signature} $\sig(\Kmc)$ is the set of all concept, role, and individual names $\Inds{\Kmc}$ occurring in it.
A \emph{conjunctive query (CQ)} $\q(\x)$ is an expression of the
form $\exists \y.\,\phi(\x,\y)$, where $\phi(\x,\y)$ is a conjunction of 
atoms~$A(t)$ or~$r(s,t)$ and $s,t$ are variables or constants. 
%
The variables in $\x$ are called \emph{answer
variables} and $\y$ are the \emph{existentially quantified variables}.
If $\q(\x)$ contains only a single unary atom, it is called \emph{instance query} (IQ).
If $\x$ is empty, then $\q(\x)$ is called \emph{Boolean}.
Note that ABox facts are a special case of Boolean CQs with only one atom
and no variables.
%
A tuple $\answer$ of constants from $\Amc$ is a \emph{certain answer} to
$\q(\x)$ over a KB $\Kmc$, written $\Kmc\models\q(\answer)$, if every model of $\Kmc$ satisfies the sentence $\q(\answer)$.
%
We may write $A(x) \in \q$ to indicate that $A(x)$ is an atom in~$\phi$.
%
A \emph{union of CQs} (UCQ) is a disjunction of CQs \added{sharing the same 
answer variables}.  
A CQ $\q(\x)$ is \emph{UCQ-rewritable} over a TBox~\Tmc if there exists a 
\emph{UCQ} $\q_\Tmc(\x)$ such that, 
for every ABox~\Amc and tuple~\a, $\Tmc\cup\Amc\models\q(\a)$ iff $\Amc\models\q_\Tmc(\a)$.
This is the case, for example, for all CQs over \DLLiteR TBoxes~\cite{JAR-2007}.
%
Since we consider proofs for a given, fixed answer~\a, we consider only the Boolean CQ $\q(\a)$, which we denote in the following simply as~\q.

\paragraph{Proofs.}
\added{Following the formal framework}
in~\cite{LPAR23:Finding_Small_Proofs_for,ABB+-DL20,ABB+-CADE21}, we \added{view} 
proofs in a logic \Lmc as finite directed hypergraphs $(V,E,\ell)$
where each vertex $v\in V$ is labeled by an \Lmc-sentence $\ell(v)$, and every 
hyperedge is of the form $(S,d)\in E$ 
the \added{finite} set $S\subseteq V$ being the \emph{premises} and $d\in V$ the \emph{conclusion}, which we may depict as
\begin{center}
 \AXC{$p$}
 \AXC{$p\to q$}
 \BIC{$q$}
 \DP
 \quad
 or
 \quad
 \scalebox{0.8}{
 \begin{tikzpicture}[
  scale=0.8,
  baseline=(x3.north),
  block/.style={
  draw,
  rounded rectangle,
  minimum width={width("$A \sqsubseteq \exists r.(\exists r.A)$")-5pt},
  inner sep=2pt,
  },
  tbox/.style={block,text=gray},
  he/.style={rounded corners=5pt,->}]
  \path (-1,3.5) node[block] (x1) {$p$}
       (5,3.5) node[tbox] (x2) {$p\to q$}
       (2,2.85) node[block] (x3) {$q$};
  \draw[he] (x1) -- ($(x1)!0.5!(x2)$)-- (x3);
  \draw[he] (x2) -- ($(x1)!0.5!(x2)$)-- (x3);
  \end{tikzpicture}
  }
\end{center}
\added{We call these edges also \emph{inferences}.}
Proofs can be found by looking at derivation structures. 
Formally, a \emph{derivation structure} over a KB~\Kmc is a possibly infinite hypergraph as above in 
which each 
inference $(S,d)$ is \emph{sound}, \added{that is}, the labels of $S$ logically entail the label of $d$,
%
and every leaf \added{(vertex without incoming edges)} is labeled by an element
of~\Kmc.
A \emph{proof} for an entailment $\Kmc\models\eta$
is a \added{finite} derivation structure that (i) is acyclic, (ii) has 
exactly one sink
(the conclusion),
which is labeled by the goal sentence~$\eta$, and (iii) in which each vertex $v$ is the conclusion of at 
most one 
hyperedge $(S,v)$.
The \emph{size} of a proof is the number of its vertices, and the \emph{tree size} is the size of its tree 
unraveling, starting from the sink.
%

Proofs are usually generated based on a calculus or some reasoning system. 
\added{This is formalized by the} notion of a 
deriver, which, for a given entailment $\Kmc\models\eta$, generates a derivation structure in which different 
possible proofs 
can be found.
Formally, a \emph{deriver} $\R$ for a
logic $\Lmc$ \added{is a function that} takes as input an $\Lmc$-theory $\Kmc$ and an
$\Lmc$-sentence~$\eta$, and returns a derivation structure
$\R(\Kmc,\eta)$ over $\Kmc$ that describes all inference
steps that $\R$ could perform in order to derive $\eta$ from $\Kmc$.
This structure is not necessarily computed explicitly, but can
be accessed through an oracle (in practice, this corresponds, for example, to
checking whether an inference conforms to a calculus).

\begin{remark} We argue that we can make some simplifying assumptions on
the shape of \HornALCHOI rules.
  \begin{enumerate}[label=(\alph*)]
  \item To keep constructions easier, we 
  assume \added{TBoxes} 
  to be normalized as in~\cite{DBLP:conf/ijcai/OrtizRS11,DBLP:conf/kr/CarralDK18}. 
  Such a normalization can always be performed in polynomial time by introducing fresh names as abbreviations for complex formulas and applying standard
  transformations. 
  We can transform a proof over a normalized \added{TBox} to a proof for the original non-normalized \added{TBox} by 
  (i) replacing the new names 
  with the original complex expressions, which may result in intermediate proof steps using atoms like 
  $(\exists r.A)(x)$, and (ii) possibly introducing new inference steps corresponding to normalization steps.
  This increases the size of the proofs at most polynomially, which is why we believe our results are also 
  relevant to non-normalized \added{TBoxes}. 
  \item We assume KBs to be consistent. Since for Horn DLs, $\bot$ is only useful to create inconsistencies, we 
    assume in the following that $\bot$ is never used. 
  
  \end{enumerate}
\end{remark}

\section{A Deriver Using Universal Models}\label{sec:DS4OMQA}

A distinguishing feature of Horn DLs is that every KB has a \emph{universal model} which satisfies 
exactly the Boolean CQs that are entailed by the KB. In the literature on existential rules, the term 
\emph{chase} 
refers to (different variants of) universal models~\cite{DBLP:journals/ws/CaliGL12}.
\added{Intuitively, a chase is constructed
by applying rules to facts, where fresh objects are introduced for existential quantified variables. As we illustrate in the introduction, proofs
connected to universal models can help to explain query answers. However,
because we require inferences to be sound, our framework does not allow for
an inference mechanism that introduces fresh objects. Our solution is to
provide a deriver that is sound \wrt the \emph{Skolemized} TBox, rather than
the original TBox. By Skolemizing, we eliminate existential quantification using fresh function symbols. The saturation 
of an ABox using a Skolemized TBox produces the least Herbrand model, which in turn corresponds to the \emph{Skolem chase} 
(a.k.a.\ semi-oblivious chase)~\cite{DBLP:conf/pods/Marnette09} of the original TBox.}
%
In our case, existential quantification only occurs in rules
of the form (iv) (
\added{see} Table~\ref{tab:normal-form}), which
then get transformed into $A(x) \to r(x,f(x))\land B(f(x))$ \added{where $f$ is unique for each existentially quantified variable}. Given a TBox $\Tmc$, we denote by
$\Tmc^s$ the result of Skolemizing all axioms in $\Tmc$. A universal model of $\Tmc\cup\Amc$ can then be obtained by
\enquote{applying} the rules in $\Tmc^s$ to $\Amc$ until a fixpoint is reached (which may result in an infinite
set of atoms).

\begin{figure}[tb]
\centering\begin{footnotesize}
  \begin{tikzpicture}[scale=0.8,
  block/.style={
  draw,
  rounded rectangle,
  },
  tbox/.style={
  draw=gray, text=gray,
  rounded rectangle,
  },
  he/.style={rounded corners=10pt,->}]
\path  (0,2) node[block,draw=blue] (x1) {$A(\ex{a})$}
      (7,2) node[tbox] (x2) {$A(x)\to r(x,f(x)) \land B(f(x))$}
      (3,0.8) node[block] (x3) {$B(f(\ex{a}))$}
      (8.5,0.8) node[tbox] (y) {$B(x)\to s(x,g(x)) \land A(g(x))$}
      (0.2,-1) node[block] (x3') {$r(\ex{a},f(\ex{a}))$}
      (5,-0.3) node[block] (y3) {$s(f(\ex{a}),g(f(\ex{a})))$}
      (10,-1) node[tbox] (y6) {$s(x,y) \land  r(z,x) \to E(x)$}
         (5,-1.8) node[block] (z0) {$E(f(\ex{a}))$}
         (10,-2.5) node[tbox] (z1) {$E(x) \land r(y,x) \to D(y)$}
       (5,-3.3) node[block] (x5) {$D(\ex{a})$}
      (3,-4.4) node[block] (z) { $r (\ex{a}, f(\ex{a})) \,\land\,  r (\ex{a}, f(\ex{a})) \,\land\, D(\ex{a})$}
            (3,-5.8) node[block,,fill=gray!20] (z') { $\exists x',y.\, r (\ex{a}, y) \,\land\,  r (x', y) \,\land\, D(x')$}
  ;
  \draw[he] (x1) -- ($(x1)!0.5!(x2)-(1.5,0)$) node[yshift=3pt] {\tiny{\gMOPO}} -- (x3);
  \draw[he] (x2) -- ($(x1)!0.5!(x2)-(1.5,0)$) -- (x3);
    \draw[he] (x1) -- ($(x1)!0.5!(x2)-(1.5,0)$) -- (x3');
  \draw[he] (x2) -- ($(x1)!0.5!(x2)-(1.5,0)$) -- (x3');
  \draw[he] (y) -| node[yshift=3pt] {\tiny{\gMOPO}} (y3);
  \draw[he] (x3) -| (y3);
        \draw[he] (x3') -| node[xshift=9pt,yshift=4pt] {\tiny{\gMOPO}} (z0);
    \draw[he] (y3) -- (z0);
  \draw[he] (y6) -| (z0);
      \draw[he] (x3') -- ($(x5)+(0,0.8)$) -- (x5);
    \draw[he] (z0) -- (x5);
  \draw[he] (z1) -| node[xshift=9pt,yshift=4pt] {\tiny{\gMOPO}} (x5);
   \draw[he] (x5) -| node[xshift=2pt,yshift=3pt] {\tiny{\gCONJ}} (z);
  \draw[he] (x3') -- ($(z)+(0,0.8)$) -- (z);
  \draw[he] ($(x3')-(0,0.33)$) -- ($(z)+(0,0.8)$) -- (z);
    \draw[he] (z) -- (z') node[xshift=8pt,yshift=11pt] {\tiny{\gEXISTS}} ;
  \end{tikzpicture}
  \end{footnotesize}
  \caption{A Skolemized proof for the example (colors are used for the ease of reading)}\label{fig:ex-skolem}
\end{figure}
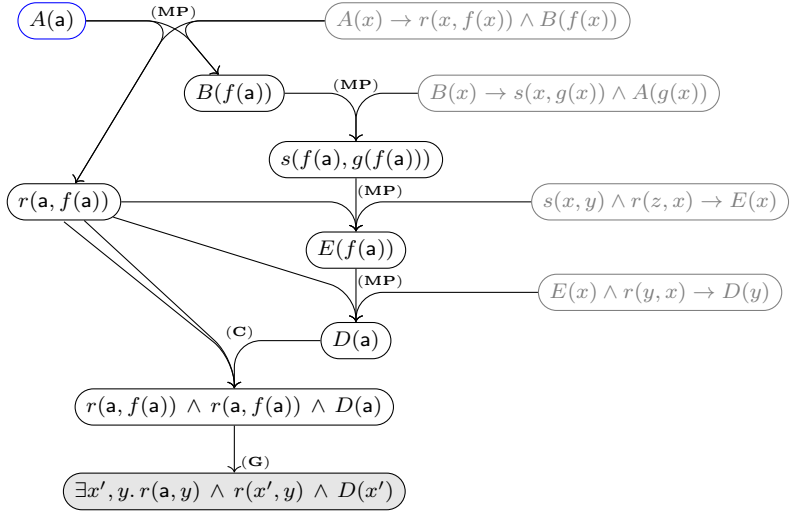

In the following, let $\Tmc\cup\Amc$ be a KB in some DL \Lmc and $\q$ a Boolean CQ with $\Tmc\cup\Amc\models\q$, which we want to explain.
For this, we define an appropriate deriver over the extended logic \Lcq, 
which contains the results of Skolemizing the rules in Table~\ref{tab:normal-form} as well as all 
Boolean CQs. 
%
%
To provide good explanations, inferences should be simple, \ie involve only small modifications of the premises.
For TBox entailment,
in~\cite{LPAR23:Finding_Small_Proofs_for,ABB+-DL20,ABB+-CADE21}, we considered
derivers based on the inference schemas used by
consequence-based reasoners.
To obtain proofs for CQs, we present the deriver~\Rsk, which inspired by the approach from~\cite{DBLP:conf/otm/BorgidaCR08} and mainly operates on ground 
CQs that may use Skolem terms, but no existential quantification.
%
%
%
%
Since ground atoms do not share variables, we mainly need to consider inferences on single atoms, which allows for 
fine-grained proofs (see Fig.~\ref{fig:ex-skolem}).
Only at the end we need to compose atoms to obtain the desired CQ~$\q$.

\begin{figure}[tb]
  \begin{framed}
    \centering
    \def\defaultHypSeparation{\hskip .5em}
    \AXC{$\alpha_1(\vec{t}_1)$}
    \AXC{$\dots$}
    \AXC{$\alpha_n(\vec{t}_n)$}
    \AXC{$\psi(\vec{y},\vec{z})\to\,\chi(\vec{z})$}
    \RL{\gMOPO}
    \QIC{$\beta(\vec{s})$}
    \DP
    \qquad
    \AXC{$\alpha(\vec{t})$}
    \AXC{$t_1=t_2$}
    \RL{\gEQUAL}
    \BIC{$\alpha(\vec{t})[t_1\mapsto t_2]$}
    \DP
    \\
    \AXC{$\alpha_1(\vec{t}_1)$}
    \AXC{$\dots$}
    \AXC{$\alpha_n(\vec{t}_n)$}
    \RL{\gCONJ}
    \TIC{$\alpha_1(\vec{t}_1)\land\dots\land\alpha_n(\vec{t}_n)$}
    \DP
    \qquad
    \AXC{$\phi(\vec{t})$}
    \RL{\gEXISTS}
    \UIC{$\exists\vec{x}.\phi(\vec{x})$}
    \DP
  \end{framed}
  \vspace*{-\baselineskip}
  \caption{Inference schemas in \Rsk (modus ponens, equality, conjunction, generalization).}
  \label{fig:dsk}
\end{figure}
The inference schemas of~\Rsk are shown in Fig.~\ref{fig:dsk}.
In~\gMOPO, $\alpha_i(\vec{t}_i)$ and $\beta(\vec{s})$ are ground atoms, $\psi(\vec{y},\vec{z})\to\,\chi(\vec{z})$ is a Skolemized rule from~$\Tmc^s$, and there must be a substitution~$\pi$ such that
$\pi(\psi(\vec{y},\vec{z}))=\{\alpha_1(\vec{t}_1),\dots,\alpha_n(\vec{t}_n)\}$
and $\beta(\vec{s})\in\pi(\chi(\vec{z}))$.
\gEQUAL deals with equalities $t_1=t_2$ by copying atoms $\alpha(\vec{t})$ that
contain~$t_1$ or~$t_2$ (we consider $=$-atoms to be symmetric). We only apply
\gEQUAL to replace top-level terms, not nested terms. Replacing also nested
terms might be logically sound, but would not improve the readability of the
proof, and is also not needed for completeness.
To complete the proof, \gCONJ combines several ground atoms into a conjunction, and \gEXISTS generalizes ground terms to 
variables in order to produce the final CQ (see Fig.~\ref{fig:ex-skolem}).
Note that the same atom can be used several times as a premise for~\gMOPO
or~\gCONJ, which then however results in a \emph{double connection} as in Fig.~\ref{fig:ex-skolem} for $r(\ex{a},f(\ex{a}))$. 
Consequently, the premise (and the subproof above it) would be duplicated in the tree unraveling of the proof.
%
%
%

\begin{definition}\label{def:cq-ds-skolem}
  $\Rsk(\Tmc^s\cup\Amc,\q)$ is \added{an infinite} derivation structure over $\Tmc^s\cup\Amc$ with \added{vertices for the} axioms in $\Tmc^s\cup\Amc$ and all Boolean CQs over $\sig(\Tmc^s\cup\Amc)$, and hyperedges for all possible instances of~\gMOPO, \gEQUAL, \gCONJ, and~\gEXISTS over these vertices.%
  \footnote{This derivation structure is uniquely determined except for the names of the vertices, which are irrelevant for our purposes since we use only their labels.}
  An \emph{(admissible) proof in $\Rsk(\Tmc^s\cup\Amc,\q)$} is a proof of $\Tmc^s\cup\Amc\models\q$ that has a label-preserving homomorphism into this derivation structure.
\end{definition}

It is easy to check that these inferences are sound.
Moreover, they are also complete, \ie if $\Tmc\cup\Amc\models\q$ holds, then there exists a proof for it (w.r.t.~$\Tmc^s$).
To see this, observe that we closely follow the (Skolem) chase construction for existential rules~\cite{DBLP:journals/ws/CaliGL12,DBLP:conf/pods/Marnette09}, where \gMOPO corresponds to standard chase steps, and \gEQUAL can be seen as merging domain elements in case of equalities (\gCONJ and \gEXISTS are only relevant to obtain the final CQ).
The resulting model~$M$ is universal, which means that
$\Tmc\cup\Amc\models\q$ implies $M\models\q$, which, in turn, shows that
there must be a proof in $\Rsk(\Tmc^s\cup\Amc,\q)$.
%



\section{Finding Good Proofs in \Rsk}

We are interested in the worst-case complexity of computing good
proofs with our deriver \Rsk.
In the following, we denote by $\msize(\p)$ ($\mtree(\p)$) the (tree) size of a proof~\p.
In addition, we consider the \emph{domain size} $\mdomain(\p)$, which is defined
as the number of ground terms appearing in~\p.
We consider the following decision problem $\OP\sk(\L,\mx)$ for some
DL~\L and measure
$\mx\in\{\msize,\mtree,\mdomain\}$:
given an \L-KB $\Tmc\cup\Amc$, a (Boolean) CQ~$\q$ such that $\Tmc\cup\Amc\models\q$,
and a natural number~$n>1$ encoded in binary\footnote{Unary encoding of~$n$ would make the problem much easier due to imposing a small (polynomial) upper bound on the (domain/tree) size of proofs. Hence, binary encoding puts more emphasis on the impact of the KB and the query on the decision problem.}, is there a proof~\p for $\q$ in
$\Rsk(\Tmc^s\cup\Amc,\q)$ with $\mx(\p)\le n$? 
%

To better isolate the complexity of finding small proofs from that of
query answering, we assume $\Tmc\cup\Amc\models\q$ as prerequisite,
which fits the intuition that users would request an explanation only after
they know that $\q$ is entailed.
Similarly to Lemma~7 in~\cite{ABB+-CADE21}, instead of looking for arbitrary
proofs and
homomorphisms into the derivation structure (
see Def.~\ref{def:cq-ds-skolem}), one can restrict the search to
\emph{subproofs}\footnote{
see Appendix~\ref{sec:app-def} for definitions.
} 
of $\Rsk(\Tmc^s\cup\Amc,\q)$, which we will often do
implicitly.
\added{
\begin{lemma}\label{lem:homo}
For any measure $\mx\in \{\msize,\mtree,\mdomain\}$, if there is an admissible proof \p \wrt $\Rsk(\Tmc^s\cup\Amc,\q)$ with
	$\mx(\p)\le n$, then there exists a subproof $\p'$ of
	$\Rsk(\Tmc^s\cup\Amc,\q)$ for $\Tmc^s\cup\Amc\models\q$ with $\mx(\p')\le n$.
\end{lemma}
Since domain size also satisfies the preconditions of Lemma~7 in~\cite{ABB+-CADE21}, the statement of Lemma~\ref{lem:homo} can be shown similarly.}

\subsection{The Data Complexity of Finding Good Proofs}\label{ssec:data-complexity}
It is common in the context of OMQA to distinguish between \emph{data complexity}, where only 
the data varies, and
\emph{combined complexity}, where also the influence of the other inputs is taken into account.
This raises the question whether the bound~$n$ is
seen as part of input for the data complexity or not. It turns out that
fixing~$n$
trivializes the data complexity, because then $n$ also fixes the set of
relevant ABoxes modulo isomorphism, so that the problem can be reduced to UCQ
entailment.

\begin{restatable}{theorem}{TheACzero}\label{th:TheACzero}
For a constant~$n$, any \L, and any $\mx\in\{\msize,\mtree,\mdomain\}$, 
$\OP\sk(\L,\mx)$ is in~$\ACzero$ in data complexity.
\end{restatable}

One may argue that, since the size of the proof depends on \Amc, the bound
$n$ on the proof size should be considered part of the input as well. Under this
assumption, our decision problem is not necessarily in \ACzero anymore. For
example, consider the \EL TBox $\{ \exists r.A\sqsubseteq A \}$ and
$q(x)\leftarrow A(x)$. For every $n$, there is an ABox
$\Amc$ such that $A(a)$ is entailed by a sequence of~$n$ role atoms, and thus needs a proof of
size at least~$n$. Deciding whether this query admits a bounded proof is thus
as hard as deciding whether it admits an answer at all in~$\Amc$, \ie
$\PTime$-hard~\cite{DBLP:conf/dlog/Rosati07}. However, the problem stays in
$\ACzero$ for DLs over which CQs are UCQ-rewritable, \eg
\DLLiteR~\cite{JAR-2007}, because the number of (non-isomorphic) proofs that we need to consider is bounded by the size of the rewriting, which is constant in data complexity.

\begin{restatable}{theorem}{TheACzeroRewritable}\label{th:TheACzeroRewritable}
For any $\mx\in\{\msize,\mtree,\mdomain\}$ and any \L such that all CQs are UCQ-rewritable over \L-TBoxes, $\OP\sk(\L,\mx)$ is in \ACzero in data complexity.
\end{restatable}

\subsection{Finding Good Proofs with Lightweight Ontologies}\label{sec:results}

We now consider the combined complexity of our problems for \DLLiteR and \EL.
In~\cite{LPAR23:Finding_Small_Proofs_for,ABB+-CADE21}, we
established general upper bounds for finding proofs of bounded size.
These results depend only on 
the size
of the derivation structure obtained 
for the given input.
%
However, $\Rsk$ does not produce finite derivation structures since there can be Skolem terms of arbitrary nesting depth.
Nevertheless, proofs cannot be infinite, and therefore we first study how large proofs in \Rsk can get in the worst case.
In particular, for \EL one can enforce proofs that are binary trees of polynomial depth, and therefore of exponential size.

\begin{restatable}{lemma}{LemSizeLowerLight}\label{lem:size-lower-light}
  One can construct a TBox $\Tmc_{\Lmc,n}$ in time polynomial in~$n$ such that
  $\Tmc_{\Lmc,n}\cup\{A(a)\}\models B(a)$, but every proof 
  \added{of the entailment} is of (domain/tree) size
  \begin{enumerate}
    \item polynomial in $n$ for $\Lmc=\DLLiteR$,
    \item exponential in $n$ for $\Lmc=\EL$.
  \end{enumerate}
Moreover, there exists an \EL-TBox~$\Tmc$ for which one can construct an ABox
$\Amc_{n}$ in time polynomial in~$n$ such that
$\Tmc\cup\Amc_{n}\models A(a)$, but every proof 
\added{of} it is of \added{a} tree size exponential in~$n$.
\end{restatable}

To obtain matching upper bounds, we can bound the number
of relevant Skolem terms in~\Rsk by investigating which part of the universal model is necessary to satisfy the query~$\q$.
%

\begin{restatable}{lemma}{LemSizeUpperDLLiteEL}\label{lem:size-upper-dl-lite-el}
    For any CQ entailment 
    $\Tmc\cup\Amc\models\q$, there exists a proof of
    \begin{enumerate}
        \item (domain/tree) size polynomial in $|\Tmc\cup\Amc|$ and $|\q|$ if $\Lmc=\DLLiteR$,
        \item (domain) size exponential in $|\Tmc|$ and $|\q|$ and polynomial in~$|\Amc|$ if $\Lmc=\EL$,
        \item tree size exponential in $|\Tmc\cup\Amc|$ and $|\q|$ if $\Lmc=\EL$.
    \end{enumerate}
\end{restatable}

This immediately allows us to show some generic upper bounds by guessing proofs up to the specified sizes.

\begin{restatable}{theorem}{ThmDLLiteELUpperBounds}\label{thm:dllite-el-upper-bounds}
  For any $\mx\in\{\msize,\mtree,\mdomain\}$,
  $\OP\sk(\EL,\mx)$ is in $\NExpTime$
  and
  $\OP\sk(\DLLiteR,\mx)$ is in~\NP.
\end{restatable}

In some cases, we can show matching lower bounds via reductions from the Boolean
query entailment problem. Using Lemma~\ref{lem:size-upper-dl-lite-el}, we can
find an upper bound $n$ for any proof showing $\Kmc\models q$ provided that it
holds. To satisfy the prerequisites of $\OP\sk$, we then extend $\Kmc$ by a second KB $\Kmc'$ in which $\Kmc'\models q$,
but only with a proof \emph{larger} than $n$.

\begin{restatable}{theorem}{ThmDLLiteLower}\label{thm:dllite-lower}
  For $\mx\in \{\msize,\mtree\}$, $\OP\sk(\DLLiteR,\mx)$ is \NP-hard.
\end{restatable}
%
%
%
%
%
To obtain tractability, we 
\added{can} restrict the shape of the
query. The \emph{Gaifman graph} of a query $\q$ is the undirected
graph
that uses the terms of $\q$ as nodes and has an edge between terms
occurring together in an atom.
%
A query is \emph{tree-shaped} if its Gaifman graph is a tree.
We can exploit this structure to deterministically explore in polynomial
time all relevant proofs
of minimal tree size over \DLLiteR KBs.

\begin{restatable}{theorem}{LemPTimeCombined}\label{th:ptime-combined}
 Given a \DLLiteR KB $\Tmc\cup\Amc$ and a tree-shaped query $\q$, 
one can compute in polynomial time a proof of minimal
tree size in $\Rsk(\Tmc^s\cup\Amc,\q)$.
\end{restatable}
The central property used in the proof 
of Theorem~\ref{th:ptime-combined} is that
for tree size
%
the proof of each atom in $\q$ is counted
separately, even if two atoms are proven in the same way.
Since \mdomain and \msize do not exhibit this redundancy, we can show that the corresponding decision problems are already \NP-hard for tree-shaped queries, and even \emph{without a TBox}, via reductions from the propositional satisfiability problem.

\begin{restatable}{theorem}{TheNPHardModified}\label{the:np-hard-trees}
  Let \Lmc be an arbitrary DL and $\mx\in\{\msize,\mdomain\}$. For tree-shaped CQs,
  $\OP\sk(\L,\mx)$ is \NP-hard.
\end{restatable}


For \EL, we can similarly show improved complexity bounds for the case of tree
size, where the lower bounds are obtained using the same idea as for
Theorem~\ref{thm:dllite-lower}, however this time using the exponential bound on
the tree size from Lemma~\ref{lem:size-upper-dl-lite-el}.
%
%
\begin{restatable}{theorem}{TheELTreeUpperBounds}\label{the:el-bounds}
%
$\OP\sk(\EL,\mtree)$ is \NP-complete in combined, and in \PTime
in data complexity. For IQs, the problem is \PTime-complete in combined complexity.
    
\end{restatable}
%


\subsection{Finding Good Proofs with Expressive
Ontologies}\label{subsec:proof-bounds}

We continue our journey towards more expressive DLs.
First, we establish a more expressive counterpart of
Lemma~\ref{lem:size-lower-light}. This time, we can even enforce trees of
exponential depth, by implementing a binary counter using
concept names for the different bit values. To produce the entailment, the
proof has to increment the counter all the way to the maximum value, and do so
on every branch of a binary tree, which gives us the desired lower bound.
\begin{restatable}{lemma}{LemSizeLower}\label{lem:size-lower}
    One can construct a $\HornALC$-TBox $\Tmc_{\Lmc,n}$ 
    in time polynomial in~$n$ such that
    $\Tmc_{\Lmc,n} \cup\{A(a)\}\models B(a)$, but every proof 
    \added{of the entailment} is of (domain/tree) size 
    \added{doubly} exponential in~$n$. 
\end{restatable}

In the case of (domain) size, we can also find a matching upper bound. The
general idea is using a kind of type construction. Intuitively, we identify
the terms occurring the proof based on the predicates they occur in. Because
there are at most exponentially many possibilities for this, we can bound the
nesting depth of Skolem terms by an exponential, which gives a double
exponential bound on domain size and size.
For
tree size, this is not so straightforward, and we leave the
exact bounds for future work.
\begin{restatable}{lemma}{LemSizeUpper}\label{lem:size-upper}
    For any CQ entailment
    $\Tmc\cup\Amc\models\q$ with $\Tmc$ being a $\HornALCHOI$-TBox, there
exists a proof of
    (domain) size double-exponential in $\Tmc$ and polynomial in $\Amc$.  
\end{restatable}
In contrast to Lemma~\ref{lem:size-upper-dl-lite-el}
for \DLLite and \EL, we cannot use Lemma~\ref{lem:size-upper} to reduce
$\OP\sk(\HornALCHOI,\m)$ to query entailment in $\HornALCHOI$ since a
double exponential bound cannot be expressed using only polynomially many bits.
On the positive side, the fact that the bound~$n$ is encoded in binary means
that for $\OP\sk(\HornALCHOI,\m)$, we do not need to consider proofs of more
than exponential
size, which gives us a \NExpTime upper bound for~$\msize$; for~$\mdomain$ it holds as well since there are exponentially many facts
 over $\sig(\Tmc^s\cup\Amc)$ with a domain bounded by $n$. Using a technique from~\cite{ABB+-CADE21}, we can even improve
this to \PSpace in the case of~\mtree.
%

\begin{restatable}{theorem}{ThHornSROIQ}\label{th:HornSROIQ}
$\OP\sk(\HornALCHOI,\mx)$
is in \NExpTime for $\mx\in\{\msize,\mdomain\}$, and in
\PSpace for $\mx=\mtree$.
%
\end{restatable}


For \msize, we are able to improve this complexity even further using a more
involved technique. The idea is to virtually construct the proof from
\emph{proof segments} which are
represented using tuples of the form $\tup{t,\In,\Out,\Size}$, where $t$ is a
term, $\In$ and $\Out$ are sets of atoms of restricted shape that may use a
placeholder~$\_$ to represent \emph{relative Skolem terms}, and
$\Size$ is an integer.
Intuitively, such a tuple tells us that it is possible to derive $\Out$
from $\In$ using a proof of size at most $\Size$. $t$~may
optionally store what the placeholder~$\_$ stands for, provided that this is
relevant for the query answer. We impose additional syntactic restrictions to
ensure that
there can be at most exponentially many such tuples. The decision
procedure starts from a set of initial proof segments that correspond to proofs
of polynomial size, and then step-wise aggregates proof segments
to represent larger proofs, with the concise tuple representation making sure
that there can be at most exponentially many such operations. We can thus prove
the following theorem.

The main observation underlying this algorithm is that \HornALCHOI rules can
only increase or decrease the nesting depth of a term by at most 1, while we can
assume that \gEQUAL only replaces terms by constants. This introduces a kind of
locality to proofs that allows us to decompose proofs in the way that is
required by our method. Since for logics with number restrictions (such as
\HornALCHOIQ), this locality assuption failed, we did not consider such logics
yet in our investigations.

\begin{restatable}{theorem}{ThHornSROIQInstance}\label{th:HornSROIQ-IQ}
 $\OP\sk(\HornALCHOI,\msize)$ is in \ExpTime.
\end{restatable}



\section{Directly Deriving CQs}\label{sec:cq}

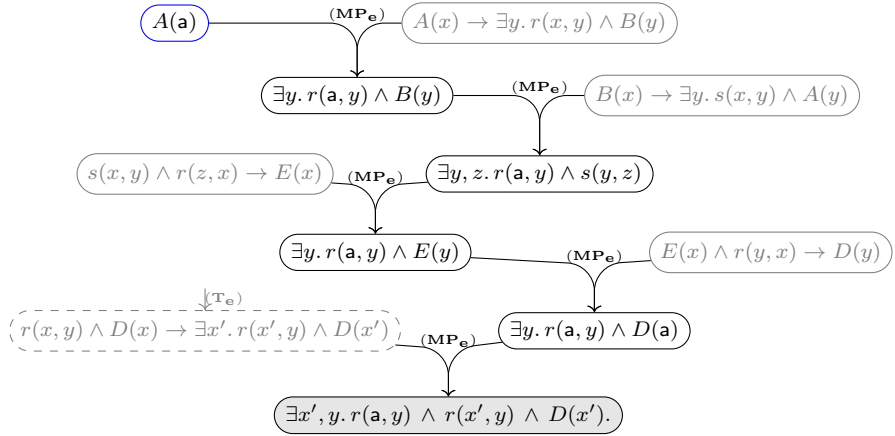
\begin{figure}[tb]
  \centering\begin{footnotesize}
  \begin{tikzpicture}[scale=0.8,
  block/.style={
  draw,
  rounded rectangle,
  },
  tbox/.style={
  draw=gray, text=gray,
  rounded rectangle,
  minimum width={width("$A \sqsubseteq \exists r.(\exists r.A \sqcap B)$")-5pt},
  minimum height=1.8em,
  },
  he/.style={rounded corners=10pt,->}]
  \path (-1,1.8) node[block, draw=blue] (x1) {$A(\ex{a})$}
      (5,1.8) node[tbox] (x2) {$A(x)\to\exists y.\,r(x,y) \land B(y)$}
      (2,0.6) node[block] (x3) {$\exists y.\,r(\ex{a},y) \land B(y)$}
      (8,0.6) node[tbox] (y) {$B(x)\to\exists y.\,s(x,y) \land A(y)$}
      (5,-0.7) node[block] (y3) {$\exists y,z.\,r(\ex{a},y) \land s(y,z) $}
          (-0.5,-0.7) node[tbox] (x4) {$s(x,y) \land  r(z,x) \to E(x)$}
         (2.3,-2) node[block] (z0) {$\exists y.\,r(\ex{a},y) \land E(y)$}
         (8.8,-2) node[tbox] (z1) {$E(x) \land r(y,x) \to D(y)$}
       (5.9,-3.3) node[block] (x5) {$\exists y.\,r(\ex{a},y) \land D(\ex{a})$}
        (-0.5,-3.3) node[tbox, dashed] (x6) {$r(x,y) \land D(x)\to \exists x'.\,r(x',y) \land D(x')$}
      (3.5,-4.7) node[block,,fill=gray!20] (z) { $\exists x',y.\, r (\ex{a}, y) \,\land\,  r (x', y) \,\land\, D(x').$}
  ;
  \draw[he] (x1) -- ($(x1)!0.5!(x2)$) node[yshift=3pt] {\tiny{\MOPO}} -- (x3);
  \draw[he] (x2) -- ($(x1)!0.5!(x2)$) -- (x3);
  \draw[he] (y) -- ($(y)!0.5!(x3)$) node[yshift=3pt] {\tiny{\MOPO}} -- (y3);
  \draw[he] (x3) -- ($(y)!0.5!(x3)$) -- (y3);
  %
  %
    \draw[he] (x4) -- ($(z0)+(0,1.1)$) node[yshift=3pt] {\tiny{\MOPO}} -- (z0);
  \draw[he] (y3) -- ($(z0)+(0,1.1)$) -- (z0);
      \draw[he] (z0) -- ($(x5)+(0,1.1)$) node[yshift=3pt] {\tiny{\MOPO}} -- (x5);
  \draw[he] (z1) -- ($(x5)+(0,1.1)$) -- (x5);
   \draw[he] (x5) -- ($(z)+(0,1.1)$) node[yshift=3pt] {\tiny{\MOPO}}  -- (z);
  \draw[he] (x6) -- ($(z)+(0,1.1)$) -- (z);
  \draw[he, gray] ($(x6)+(0,0.7)$) -- (x6) node[yshift=12pt,xshift=7pt] {\tiny{\TAUT}} ;
  \end{tikzpicture}
  \end{footnotesize}
  \caption{A CQ proof for the example}
  \label{fig:ex-existential}
\end{figure}
In addition to connecting proofs to the universal model, \Rsk has the advantage
that we can work with single atoms, which makes it easy to see
how the existential rules are applied. However, the resulting proofs are not
sound \wrt the ontology~\Tmc, but only \wrt the Skolemized rules~$\Tmc^s$.
In order to be sound w.r.t.~\Tmc, inspired
by~\cite{Stefanoni-11,DBLP:conf/ekaw/CroceL18}, we can work directly with
Boolean CQs (see Fig.~\ref{fig:ex-existential}). 
%
%
\added{Because these proofs do not work on universal models, and do not refer
to introduced individuals directly, domain size is irrelevant in this
setting, which is why we do not consider it here.}

The corresponding inference schemas are shown in Fig.~\ref{fig:dcq}.
Now, the basic inference \MOPO matches the left-hand side of a rule in~\Tmc to
part of a CQ and then
replaces it by (part of) the right-hand side.
%
%
%
%
Additionally, we allow to keep the replaced atoms from the original CQ.
Again, \MOPO is admissible only if there exists a substitution $\pi$ such
that $\pi(\psi(\vec{y},\vec{z}))\subseteq\phi(\vec{x})$, and then
$\rho(\vec{w})$ is the result of replacing \emph{any subset of}
$\pi(\psi(\vec{y},\vec{z}))$ in $\phi(\vec{x})$ by \emph{any subset of}
$\pi(\chi(\vec{z},\vec{u}'))$, where the variables $\vec{u}$ are renamed into
new existentially quantified variables $\vec{u}'$ to ensure that they are
disjoint with $\vec{x}$.
%
%
To duplicate variables, we introduce tautological rules such as $P(x,z)
\to \exists z'.\, P(x,z')$ via~\TAUT, which yields $\exists z,z'.\,{P}(\ex{b},z)\land
{P}(\ex{b},z')$ when combined with $\exists z.\,{P}(\ex{b},z)$ using~\MOPO.
%
%
The remaining inference schemas are similar to the ones in~\Rsk, but not
restricted to ground atoms.
For \CONJ, we rename the variables $\vec{y}$ to $\vec{u}$ to avoid overlap
with~$\vec{x}$.
%
\begin{figure}[tb]
  \begin{framed}
    \centering
    \AXC{$\exists\vec{x}.\,\phi(\vec{x})$}
    \AXC{$\psi(\vec{y},\vec{z})\to\exists\vec{u}.\,\chi(\vec{z},\vec{u})$}
    \RL{\MOPO}
    \BIC{$\exists\vec{w}.\rho(\vec{w})$}
    \DP
    \qquad
    \AXC{\vphantom{$\exists\x.\,\phi(\x)$}}
    \RL{\TAUT}
    \UIC{$\phi(\vec{x},\vec{y})\to\exists\vec{x}.\,\phi(\vec{x},\vec{y})$}
    \DP

    \bigskip

    \AXC{$\exists\x.\,\phi(\x)\land t_1=t_2$}
    \RL{\EQUAL}
    \UIC{$\exists\x.\,\phi(\x)[t_1\mapsto t_2]$}
    \DP
    \quad
    \AXC{$\exists\vec{x}.\,\phi(\vec{x})$}
    \AXC{$\exists\vec{y}.\,\psi(\vec{y})$}
    \RL{\CONJ}
    \BIC{$\exists\vec{x},\vec{u}.\phi(\vec{x})\land\psi(\vec{u})$}
    \DP
    \quad
    \AXC{$\exists\vec{x}.\,\phi(\vec{x},\vec{a})$}
    \RL{\EXISTS}
    \UIC{$\exists\vec{x},\vec{y}.\,\phi(\vec{x},\vec{y})$}
    \DP  
  \end{framed}
  \vspace*{-\baselineskip}
  \caption{Inference schemas for \Rcq.}
  \label{fig:dcq}
\end{figure}
%
%

%
%
%

\begin{definition}[CQ Deriver]\label{def:cq-ds}
  The derivation structure $\Rcq(\Tmc\cup\Amc,\q)$ is defined similarly to~\Rsk, but using \MOPO, \TAUT, \EQUAL, \CONJ, and~\EXISTS.
  We also define $\OP\cq$ analogously to $\OP\sk$.
\end{definition}

Proofs obtained through \Rcq are sound \wrt the original KB and do not depend
on the notion of universal model. However, these proofs are more complex since
vertices are not labeled with single atoms anymore, making it harder to
understand how a rule is applied in case of an \MOPO inference.
Indeed, verifying individual \MOPO steps is even \NP-hard, since it requires to
match one set of atoms into another, which is equivalent to database query
answering~\cite{AbHV-95}.
This could potentially be solved by also showing the substitutions
corresponding to these inference steps to the user, but this would lead to even
more information being included in a single inference step. In general, we
believe that except for the advantage of soundness, proofs based on CQs
are less helpful for explaining query answers to users.
In case users still prefer an inference system that is sound \added{\wrt the
original TBox rather than just the Skolemized version}, we observe that
it is not hard to translate proofs based on \Rsk into proofs in
\Rcq and vice versa.

\begin{restatable}{theorem}{ThmTransformation}\label{thm:transformation}
Any proof~$\p$ in $\Rcq(\Tmc\cup\Amc,\q)$ can be transformed into
a proof in $\Rsk(\Tmc^s\cup\Amc,\q)$ in time polynomial in the
sizes of~$\p$ and~\Tmc, and conversely any proof
$\p$ in $\Rsk(\Tmc^s\cup\Amc,\q)$ can be transformed into
a proof in $\Rcq(\Tmc\cup\Amc,\q)$ in time polynomial in the
sizes of~$\p$ and~\Tmc.
The latter also holds for tree proofs.
\end{restatable}
%
This theorem also shows that this deriver is complete for query entailment since
we already know that \Rsk is complete.
However, it is not the case that \emph{minimal} proofs are equivalent for these
two derivers, \ie a minimal proof may become non-minimal after the
transformation.
Nevertheless, many of the results we have seen before also apply to \Rcq (see the
\iftr
appendix 
\else
extended version\footnotemark[\getrefnumber{link}]
\fi
and~\cite{DL22paper} for details). However, due to duplication of atoms
via~\TAUT, some results can differ (\cf Theorem~\ref{the:np-hard-trees}):
  
\begin{restatable}{theorem}{ThmNPHardCQStructure}\label{lem:np-hard-cq-structure}
  Let \Lmc be an arbitrary DL. For tree-shaped CQs, $\OP\cq(\Lmc,\msize)$ and $\OP\cq(\Lmc,\mtree)$ are \NP-hard.
\end{restatable}



\section{Conclusion}

\begin{figure}[tb]
  \centering\begin{footnotesize}
    \begin{tikzpicture}[scale=0.8,
    block/.style={
    draw,
    rounded rectangle,
    },
    tbox/.style={
    draw=gray, text=gray,
    rounded rectangle,
    },
    he/.style={rounded corners=10pt,->}]
    \path  (2,-0.3) node[block, draw=blue] (x) {$A(\ex{a})$}
        (10.5,2.3) node[tbox, dashed] (x0) {$A(x)\to r(x,f(x)) \land B(f(x))$}
        (10.5,1.3) node[tbox,dashed] (x1) {$B(x)\to s(x,g(x)) \land A(g(x))$}
        (10.5,0.3) node[tbox,dashed] (x4) {$s(x,y) \land  r(z,x) \to E(x)$}
        (10.5,-0.7) node[tbox,dashed] (x5) {$E(x) \land r(y,x) \to D(y)$}
         (4,1.7) node[tbox] (y0) {$A(x)\to r(x,f(x))$}
         (4,0.6) node[tbox] (y1) {$A(x)\to D(x)$}
              (0,-1.3) node[block] (y2) {$r(\ex{a},f(\ex{a}))$}
            (4,-1.3) node[block] (y3) {$D(\ex{a})$}
           (2,-2.4) node[block] (z3) { $r (\ex{a}, f(\ex{a})) \,\land\,  r (\ex{a}, f(\ex{a})) \,\land\, D(\ex{a})$}
        (2,-3.5) node[block,fill=gray!20] (z) { $\exists x',y.\, r (\ex{a}, y) \,\land\,  r (x', y) \,\land\, D(x')$};
      \draw[he,gray] ($(x0)-(2.67,0.2)$) -- ($(y0)+(2.6,0)$) node {$\blacksquare$} -- (y0);
    \draw[he,gray] ($(x0)-(2.5,0.3)$) -- ($(y1)+(2.3,0)$) node {$\blacksquare$} -- (y1);
    \draw[he,gray] (x1) --($(y1)+(2.3,0)$) -- (y1);
      \draw[he,gray] (x4) --($(y1)+(2.3,0)$) -- (y1);
    \draw[he,gray] ($(x5)-(2,-0.3)$) --($(y1)+(2.3,0)$) -- (y1);
    \draw[he] (x) -| (y2);
    \draw[he] (y0) -| (y2);
    \draw[he] (y1) -- ($(y3)+(0,0.8)$) -- (y3);
    \draw[he] (x) -| (y3);
    \draw[he] ($(y2)+(1,0.1)$) -- ($(y2-|z3)+(0,0.1)$) -- (z3);
    \draw[he] ($(y2)+(1,-0.1)$) -- ($(y2-|z3)+(0,-0.1)$) -- (z3);
    \draw[he] (y3) -| (z3);
    \draw[he] (z3) --(z);
    \end{tikzpicture}
    \end{footnotesize}
    \caption{A Skolemized proof for the example with hidden TBox inferences}\label{fig:ex-skolem-modular}
  \end{figure}
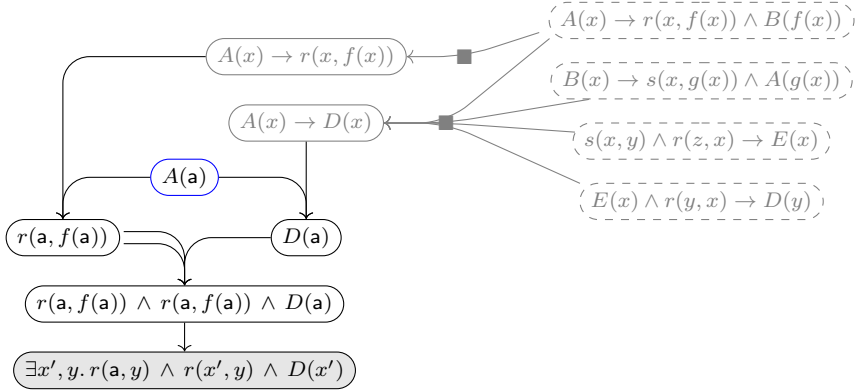
  
We have presented a general framework for generating proofs for answers to
ontology-mediated queries. The central idea is to explain the reasoning steps
that contributed to the answer by referring to a universal model. We have
also shown some initial complexity results, and intend to obtain a more precise
picture in the future.
%
%
An interesting future direction is to investigate derivers that combine TBox
and query entailment rules, \eg \Rsk plus the rules of the ELK
reasoner~\cite{DBLP:journals/jar/KazakovKS14}.
On one extreme end, one could completely hide all TBox reasoning steps, which
could result in a proof like in Fig.~\ref{fig:ex-skolem-modular}, but it would
be interesting to evaluate mixed proofs \wrt the comprehensibility of TBox- \vs
query-based inferences.
For explaining missing answers, we also want to investigate how to find
(optimal)
counter-interpretations or abduction results~\cite{DBLP:conf/ijcai/Koopmann21}.

\paragraph{Acknowledgments}
  This work was supported by the DFG grant 389792660 as part of TRR~248 -- 
  CPEC (\url{https://perspicuous-computing.science}), and QuantLA, GRK 1763
(\url{https://lat.inf.tu-dresden.de/quantla}).

\bibliographystyle{abbrv}
\bibliography{bibs}

\appendix

\newpage
\section{Additional Definitions: Hypergraphs}\label{sec:app-def}

\begin{definition}[Hypergraph]
A \emph{(finite, directed, labeled) hypergraph}~\cite{DBLP:journals/cor/NielsenAP05} is a triple $H=(V,E,\el)$, where
\begin{itemize}
\item $V$ is a finite set of \emph{vertices},
\item $E$ is a set of \emph{hyperedges}~$(S,d)$ with a tuple of \emph{source vertices} $S$
and
\emph{target vertex} $d\in V$, and
\item $\el\colon V\to \mathcal{S}_\L$ is a \emph{labeling function} that assigns sentences to
vertices.
\end{itemize}
\end{definition}
We extend the function~$\ell$ to hyperedges as follows: $\ell(S,d):=\big((\ell(s))_{s\in S},\ell(d)\big)$.

We assume that the \emph{size}~$|\eta|$ of an \L-sentence~$\eta$ is defined in
some way, \eg by the number of symbols in~$\eta$.
The \emph{size} of~$H$, denoted~$|H|$, is measured by the size of the labels of
its hyperedges: $$\card{H}:=\sum_{(S,d)\in E}\card{(S,d)}, \text{ where } \card{(S,d)}:=\card{\el(d)}+\sum_{v\in
S}\card{\el(v)}.$$
A vertex $v\in V$ is called a \emph{leaf} if it has no incoming hyperedges, \ie there is no
$(S,v)\in E$; and~$v$ is a \emph{sink} if it has no outgoing hyperedges, \ie there is no $(S,d)\in
E$ such that $v\in S$. We denote the set of all leafs and the set of all sinks in~$H$ as
$\leaf(H)$ and $\sink(H)$, respectively.

A hypergraph $H'= (V', E', \el')$ is called a \emph{subgraph} of $H= (V, E, \el)$ if $V'\subseteq V$, $E'\subseteq E$ and $\el'=\el|_{V'}$.
In this case, we also say that $H$ \emph{contains} $H'$ and write $H' \subseteq H$.
Given two hypergraphs $H_1=(V_1,E_1,\el_1)$ and $H_2=(V_2,E_2,\el_2)$ \st
$\el_1(v)=\el_2(v)$ for every $v\in V_1\cap V_2$, the \emph{union} of the two hypergraphs is
defined as $H_1 \cup H_2:=$ $(V_1\cup V_2,E_1\cup E_2, \el_1 \cup \el_2)$.

\begin{definition}[Cycle, Tree]
Given a hypergraph $H=(V,E,\el)$ and $s,t\in V$, a \emph{path} $P$ of length $q\geq 0$ in~$H$ from~$s$ to~$t$ is a sequence of vertices and
hyperedges
\[ P=(d_0,i_1,(S_1,d_1),d_1,i_2,(S_2,d_2),\dots, d_{q-1},i_q,(S_q,d_q),d_q), \]
where $d_0=s$, $d_q=t$, and $d_{j-1}$ occurs in~$S_j$ at the $i_j$-th position, for all $j$, $1\le j\le q$. If
there is such a path of length $q > 0$ in~$H$, we say that $t$ is \emph{reachable} from~$s$
in~$H$.
If $t = s$, then $P$ is called a \emph{cycle}.
The hypergraph~$H$ is \emph{acyclic} if it does not contain a cycle.
%
%
The hypergraph~$H$ is \emph{connected} if
every vertex is connected to every other vertex by a series of paths and reverse paths.

A hypergraph $H=(V,E,\el)$ is called a \emph{tree} with \emph{root} $t\in V$ if $t$ is reachable
from every vertex $v\in V\setminus\{t\}$ by exactly one path. In particular, the root is the only
sink in a tree, and all trees are acyclic and connected.
\end{definition}

\begin{definition}[Homomorphism]\label{def:homomorphism}
Let $H=(V,E,\el)$, $H'=(V',E',\el')$ be two hypergraphs.  A \emph{homomorphism} from $H$ to 
$H'$, denoted $h\colon H\rightarrow H'$, is a mapping $h\colon V\to V'$ s.t.\ for 
all~$(S,d)\in E$, one has $h(S,d):=\big((h(v))_{v\in S},\,h(d)\big)\in E'$ and, for all $v\in V$, it 
holds that~$\el'(h(v))=\el(v)$.
Such an~$h$ is an \emph{isomorphism} if it is a bijection, and its inverse, $h^-\colon H' \to H$,
is also a homomorphism.
\end{definition}

\begin{definition}[Hypergraph Unraveling]
  The \emph{unraveling} of an acyclic hypergraph $H=(V,E,\ell)$ at a vertex $v\in V$ is the tree 
  $H_T=(V_T,E_T,\ell_T)$, where $V_T$ consists of all paths in~$H$ that end in~$v$, $E_T$ 
  contains all hyperedges $((P_1,\dots,P_n),P)$ where each~$P_i$ is of the 
  form $\big(d_i,i,((d_1,\dots,d_n),d)\big)\cdot P$, and $\ell_T(P)$ is the label of the starting 
  vertex of~$P$ in~$H$.
  Moreover, the mapping $h_T\colon V_T\to V$ that maps each path to its starting vertex (and in 
  particular~$h_T(v)=v$) is a homomorphism from~$H_T$ to~$H$.
\end{definition}



\begin{definition}[Derivation Structure]\label{def:derivation-structure}
  A \emph{derivation structure} $\ds = (V, E, \el)$ over a theory~$\mathcal{U}$ is a directed, labeled hypergraph that is
  \begin{itemize}
  \item \emph{grounded}, \ie
  every leaf $v$ in~\ds is labeled by $\el(v)\in\mathcal{U}$; and
  \item \emph{sound}, \ie for all hyperedges $(S,d)\in E$, the entailment
$\{\el(s)\mid s\in S\}\models\el(d)$ holds.
\end{itemize}
\end{definition}


\begin{definition}[Proof]\label{def:proof}
  Given a sentence~$\eta$ and a theory~$\mathcal{U}$, a \emph{proof of $\mathcal{U}\models\eta$} is a 
  derivation structure $\p = (V, E,\el)$ over~$\mathcal{U}$ such that
  \begin{itemize}
    \item\label{item1:def-proof-non-redundancy} \p contains exactly one
        sink~$v_\eta\in V$, which is labeled by~$\eta$,
    \item \p is acyclic, and
    \item every vertex has at most one incoming hyperedge, \ie there exist no two hyperedges $(S_1,v),(S_2,v)\in E$ with $S_1\neq S_2$.
  \end{itemize}
A \emph{tree proof} is a proof that is a tree.
A \emph{subproof} $S$ of a hypergraph~$H$ is a subgraph of~$H$ that is a proof 
with $\leaf(S)\subseteq\leaf(H)$.
\end{definition}

The \emph{tree size} $\mtree(\p)$ of a proof~\p is defined recursively as follows:
\begin{align*}
  \mtree(v)&:=1 &&\text{for every leaf }v, \\
  \mtree(d)&:=1+\sum_{v\in S} \mtree(v), &&\text{for every }(S,d)\in E, \\
  \mtree(\p)&:=\mtree(v), &&\text{for the sink vertex }v.
\end{align*}
This recursively counts the vertices in subproofs and is equal to the size of a tree unraveling of~\p~\cite{ABB+-CADE21}.


\section{Proof Details}
\subsection{\nameref{ssec:data-complexity}}

\TheACzero*
\begin{proof}
We fix a set $\NC^n$ of $\leq n$ individual names, and collect in
$\mathfrak{A}$ all of the (constantly many) ABoxes~$\Amc$ using only names from
$\sig(\Tmc)\cup\NC^n$ for which $\Rsk(\Tmc^s\cup\Amc,\q)$ contains
a proof~\p with $\m(\p)\leq n$.
The latter can be done, \eg by breadth-first search for proofs up to (tree/domain) size~$n$; for domain size, observe that, for a given set of ground atoms~$S$, the number of possible inference steps that do not produce any new Skolem terms is bounded by a function that depends only on~$S$ and~\Tmc.
We can identify every such ABox $\Amc$ with a
CQ $q_\Amc$ obtained by replacing all individual names 
by
existentially quantified variables. We now have that, for any ABox~$\Amc$, there
exists an $\Amc'\in\mathfrak{A}$ with $\Amc\models q_{\Amc'}$ iff there
exists a proof in $\Rsk(\Tmc^s\cup\Amc,\q)$ with $\m(\p)\le n$.
Consequently, we can reduce our decision problem to deciding
whether~$\Amc$ entails the (fixed) \emph{union of conjunctive queries} (UCQ)
\[
  \bigvee_{\Amc\in\mathfrak{A}}q_\Amc.
\]
Deciding UCQ entailment without a TBox is possible in \ACzero in data complexity~\cite{AbHV-95}.
%
\end{proof}

\TheACzeroRewritable*
\begin{proof}
Let $Q$ be a UCQ that is a rewriting of $\q$ over~\Tmc, \ie such that $\Tmc\cup\Amc\models\q$ is equivalent to $\Amc\models Q$ for all ABoxes~\Amc.
For every $q'\in Q$, we
determine the minimal (tree/domain) size $n_{q'}$ for $\Tmc\cup\Amc_{q'}\models
\q$, where $\Amc_{q'}$ is obtained from $q'$ by replacing every variable
by a fresh individual name.
These ABoxes represent all possibilities of an ABox entailing $\q$ w.r.t.~\Tmc (modulo isomorphism), and hence they can be used as (a fixed number of) representatives in the search for small proofs of the entailment.
The computation of~$n_{q'}$ does not depend on the input data, and hence can be done offline via bounded search in the derivation structure.
Let $n_\textit{max}$ be the maximum of the values~$n_{q'}$. To
every $n\leq n_\textit{max}$, we assign the UCQ
 \[
  Q_n=\bigvee\{q'\in Q\mid n_{q'}\leq n\}
 \]
 Given $\Amc$ and $n$, we can now decide whether $\Rsk(\Tmc^s\cup\Amc,
\q)$ contains a proof~\p with $\m(\p)\leq n$ by 1) computing the UCQ $Q_{n'}$
for $n'=\min(n,n_\textit{max})$, and 2) checking whether $\Amc\models
Q_{n'}$.
2) is the standard query answering problem, which has \ACzero
data complexity~\cite{AbHV-95}. To see that the combined task 1)+2) can be done in \ACzero, we
note that $n_\textit{max}$ does not depend on the data, so that we only need to
process the least $\log_{n_\textit{max}}$ bits of~$n$ to determine~$n'$, which
can be done by a circuit of constant depth. 
\end{proof}

\subsection{\nameref{sec:results}}
\LemSizeLowerLight*
\begin{proof}
    For \DLLiteR, this is trivial.
    For \EL, we define 
    \begin{align*}
        \Tmc_{n,\EL}\ =\ \{&&  A& \sqsubseteq A_1,\ \ A_n\sqsubseteq B_n,\ \ B_1\sqsubseteq B\ \}\\
        && A_i &\sqsubseteq\exists r.A_{i+1}\sqcap\exists s.A_{i+1},\\
        && \exists r.B_{i+1}\sqcap\exists s.B_{i+1}&\sqsubseteq B_i \quad\quad\quad \mid\quad  0<i<n\  \}.       
    \end{align*}
    The fragment of the universal model of $\Tmc_{n,\EL}$ needed to derive $B(a)$ corresponds to a binary tree of depth $n$, and is thus exponentially large. 
    This also gives a lower bound for the (tree) proof size.
    Note that $\Tmc_{n,\EL}$ is not in normal form (\cf
Table~\ref{tab:normal-form}), but it can be transformed into normal form without
changing its size or the size of the resulting proofs more than polynomially.

For the second claim, we define
 \[
  \Amc_n=\{ A(a_0), r(a_1,a_0), s(a_1,a_0), \ldots, r(a_n, a_{n-1}),
s(a_n,a_{n-1})\},
 \]
 where $a_n=a$, and $\Tmc=\{\exists r.A\sqcap\exists s.A\sqsubseteq A\}$.
Clearly, $\Tmc\cup\Amc_n\models A(a)$, for which one has to prove each of
$A(a_i)$, $i\in\{0,\ldots,n\}$. Moreover, this relies twice on
$A(a_{i-1})$, which means that the tree unravelling of the proof will be
of size exponential in~$n$.
\end{proof}

\LemSizeUpperDLLiteEL*
\begin{proof}
For (domain) size, 
we can bound the number
of relevant Skolem terms in~\Rsk by considering only
the part of the minimal Herbrand model~$H$ that is necessary to satisfy the query~$\q$.
For example, in logics with the \emph{polynomial witness 
property}~\cite{DBLP:journals/ai/GottlobKKPSZ14}, including \DLLiteR, we know that any 
query that is entailed is already satisfied after polynomially many chase steps used to 
construct~$H$.
%
%

%
For tree size,
 \gCONJ and \gEXISTS only need to be applied once, at the very end, to produce
the query $\q$. For \DLLiteR, the remaining rule \gMOPO is such that it
always has one premise that is a CQ. Consequently, we can always construct a
proof that is composed of $\lvert\q\rvert$ linear proofs (one for each
atom), which are then connected using \gCONJ and produce the conclusion with
\gEXISTS.
As argued above, we can assume that the nesting depth
of Skolem terms is polynomially bounded by $\lvert\Tmc\rvert$, which means the
number of labels on each path is polynomially bounded by $\lvert\Tmc\rvert$ as
well.
%
%
Additionally, we can always simplify any
proof in which the same label occurs twice along some path, which means that
this polynomial bound transfers also to the depth of our proof. We obtain that
we can always find a proof of polynomial tree size.

For \EL, our proof is structured as follows: 1)~we define a
compressed derivation structure of size polynomial in $\Tmc\cup\Amc$ and $\q$,
similar as in the proof for Lemma~\ref{th:ptime-combined}, 2)~we show that proofs in this structure can be translated into proofs in $\Rsk(\Tmc^s\cup\Amc,\q)$ with the same tree size, and vice versa, and 3)~we conclude that the nesting depth of Skolem terms in such proofs can be bounded polynomially and therefore the proof size is at most exponential.

%
We first consider only instance queries $\q=A(a)$.
We replace every
Skolem term $f(t)$ by a fresh individual name $c_f$ (we ignore what is nested
under the Skolem term). The resulting theory $\Tmc^c$ has no nested terms
anymore. Our compressed derivation structure
$\mathcal{D}$ is now obtained from
$\Rsk(\Tmc^c\cup\Amc,A(a))$ by dropping applications of \gEXISTS and
\gCONJ, which are not needed since we want to derive a CQ with only one
atom $A(a)$.
$\mathcal{D}$ contains only polynomially many nodes: one for each
TBox axiom, and one for each term $A(c)$, $r(c,d)$ with $c$ and $d$ taken from
our polynomially bounded set of individual names.

Now we show how proofs~\p in the compressed derivation structure $\mathcal{D}$ can be
translated into proofs in $\Rsk(\Tmc\cup\Amc,\q)$ of the same tree size and
vice versa.
We recursively translate
$\p$ into a proof in $\Rsk(\Tmc^s\cup\Amc, A(a))$ by changing the terms
in the labels of the proof nodes, guided by the role atoms. Specifically, in
an atom of the form $r(t,c_f)$, we know that $c_f$ has to be replaced by $f(t)$.
Moreover, any inference that had $r(t,c_f)$ as premise is now adapted such that in
the other premises, $c_f$ is also replaced by $f(t)$, as well as in all
relevant predecessors of those premises. We repeat this step inductively until
all fresh individual names are replaced by Skolem terms again.
This will happen
because 1) along every branch in the proof, any fresh individual name will have
to be eliminated eventually, since the conclusion does not contain fresh
individual names, 2)  on the left-hand side of rules that correspond to \EL
axioms, if a variable in an atom $A(x)$ occurs together with another
variable, it must do so in an atom of the form $r(x,y)$. It follows that on
every branch of the proof, atoms of the form $A(c_f)$ either produce
other unary atoms with the same individual name, or eventually occur together
with an atom of the form $r(t,c_f)$ in an inference, and in both cases our
transformation will replace $c_f$ by $f(t)$. It remains to replace the TBox
rules in the proof by the original Skolemized version. The resulting tree-shaped proof is
still sound, now a proof in $\Rsk(\Tmc^s\cup\Amc, A(a))$, and has the
same tree size as~$\p$.
This transformation can also be easily done in the
other direction (simply replacing any term $f(t)$ by $c_f$).

Now consider the general case where $\q$ does not have to be an
instance query. We note that with a little modification, we can also also
transform proofs in which the final conclusion contains the fresh individual
names $c_f$: the translation procedure will then not succeed in replacing every
fresh individual name based on the role atoms in the proof, but we can then
check using the inferences introducing~$c_f$ which term would be appropriate, obtaining a proper proof in $\Rsk(\Tmc^s\cup\Amc,A(a))$.

It follows from our construction that one can always find a proof whose depth
is polynomially bounded, since the compressed derivation structure is only of
polynomial size. Consequently, the nesting depth of Skolem terms must be always
polynomially bounded as well.
We obtain the desired exponential bounds on (domain/tree) size for CQ proofs
\wrt \EL TBoxes.
Moreover, since the number of function symbols depends only on the TBox, we
obtain that the number of Skolem terms is polynomial in the number of
individual names occurring in $\Amc$, and thus polynomial in $\Amc$. Since this
also bounds the number of atoms that can occur in a proof, we also obtain
bound on the proof size that is polynomial \wrt \Amc.
\end{proof}

\ThmDLLiteLower*
\begin{proof}
We reduce the problem of query answering to the given problems. Specifically,
let $\Tmc\cup\Amc$ be a DL-Lite KB, $\q$ a query, 
and suppose we want to determine whether
$\Tmc\cup\Amc\models \q$. We know by Lemma~\ref{lem:size-upper-dl-lite-el}
that,
if $\Tmc\cup\Amc\models \q$, then there exists a proof for this
in $\Rsk(\Tmc\cup\Amc, \q)$
that
is of (tree) size at most $n:=p(\lvert{\Tmc}\rvert,\lvert\q\rvert)$, where $p$ is some
polynomial.

We now
construct a new KB $\Tmc_0\cup\Amc_0$ such that $\Tmc_0\cup\Amc_0\models
\q$, but only with a proof of (tree) size $>n$.
$\Amc_0$ is obtained from the atoms of~$\q$ by replacing every quantified variable by a fresh individual name, and each predicate~$P$ by~$P_0$.
$\Tmc_0$ contains for
every
predicate $P$ occurring in $\q$ and every $i\in\{0,\ldots,n\}$ the CI
$P_{i}\sqsubseteq P_{i+1}$, as well as $P_n\sqsubseteq P$.
Clearly,
$\Tmc_0\cup\Amc_0\models \q$, and every proof for this corresponds to a
tree of depth $n+1$, and is thus larger than $n$. Moreover, $\Tmc_0$,
$\Amc_0$, and~$n$ are all of polynomial size in the size of the input to the query answering problem.
Now set $\Tmc_1=\Tmc_0\cup\Tmc$, $\Amc_1=\Amc_0\cup\Amc$. We have
$\Tmc_1\cup\Amc_1\models \q$; however, a proof of (tree) size~$\leq n$ exists in 
$\Rsk(\Tmc_1\cup\Amc_1,\q)$ iff $\Tmc\cup\Amc\models\q$. 
\end{proof}

\LemPTimeCombined*
\begin{proof}
We construct a compressed version of $\Rsk(\Tmc\cup\Amc,\q)$ of
polynomial size.
  We introduce for every role~$R$ the individual name $b_{\exists R}$.
 Our compressed derivation structure is defined inductively as follows, where
for a role name~$P$, we identify $P^-(a,b)$ with $P(b,a)$ and $(P^-)^-$ with~$P$.
 \begin{itemize}
  \item every axiom $\alpha\in\Tmc\cup\Amc$ has a node $v_\alpha$ with
$\el(v_\alpha)=\alpha$,
  \item for nodes $v_1$, $v_2$ with $\el(v_1)=A(a)$ and $\el(v_2)=A\sqsubseteq
B$ there is an edge $(\{v_1,v_2\},v_3)$, where $\el(v_3)=B(a)$.
  \item for nodes $v_1$, $v_2$ with $\el(v_1)=P(a,b)$ and
$\el(v_2)=P\sqsubseteq Q$, there is an edge $(\{v_1,v_2\},v_3)$ with
$\el(v_3)=Q(a,b)$,
  \item for nodes $v_1$, $v_2$ with $\el(v_1)=A(a)$ and
$\el(v_2)=A\sqsubseteq\exists P$, there is an edge $(\{v_1,v_2\},v_3)$
with $\el(v_3)=P(a,b_{\exists P^-})$,
  \item for nodes $v_1$, $v_2$ with $\el(v_1)=P(a,b)$ and
$\el(v_2)=\exists P\sqsubseteq A$, there is an edge $(\{v_1,v_2\},v_3)$
with $\el(v_3)=A(a)$,
  \item for nodes $v_1$, $v_2$ with $\el(v_1)=P(a,b)$ and
$\el(v_2)=\exists P\sqsubseteq\exists Q$, there is an edge
$(\{v_1,v_2\},v_3)$
with $\el(v_3)=Q(a,b_{\exists Q^-})$.
 \end{itemize}
 Due to the conclusions with the fresh individual names, the inferences in this
compressed derivation structure are not sound, so that it is not really a
derivation structure. But because its size is polynomial, we can use the \PTime
procedure from~\cite[Lemma~11]{ABB+-CADE21} to compute for every node $v$ a \enquote{proof} of
minimal size. To use these proofs to construct a tree proof in $\Rsk(\Tmc\cup\Amc,\q)$,
we still need to match the variables in $\q$ to the constants in the
derivation structure.

For every pair $\tup{t_1,t_2}$ of terms occurring in $\q$ together in
an atom, and every possible assignment $(t_1\mapsto a_1, t_2\mapsto a_2)$ of
these terms to individual names from the compressed derivation structure, we
assign a cost $\gamma(t_1\mapsto a_1, t_2\mapsto a_2)$ that is the sum of the
minimal proof sizes for every atom in $\q$ that contains only $t_1$
and $t_2$, with these terms replaced using the assignment. We build a labeled
graph, the \emph{cost graph}, with every node a mapping from one term in
$\q$ to one constant in our compressed derivation structure, and every
edge between two nodes $(t_1\mapsto a_1)$ and $(t_2\mapsto a_2)$ labeled with
the cost $\gamma(t_1\mapsto a_1, t_2\mapsto a_2)$ (no edge if there is no edge
between $t_1$ and $t_2$ in the Gaifman graph). Every edge in the cost graph
corresponds to an edge in the Gaifman graph of $\q$, but the same edge
in the Gaifman graph can be represented by several edges in the cost graph.
We can thus transform the cost graph into a directed acyclic graph, making
sure that for every edge $(t_1\mapsto a_1, t_2\mapsto a_2)$, the edge
$(t_1,t_2)$ points from the root towards the leafs of the Gaifman graph, where
we choose an arbitrary node as the root.

  We now eliminate assignments from the cost graph starting from the leafs:
 \begin{itemize}
  \item Consider a node $(t_1\mapsto a_1)$ and for some term $t_2$,
  all outgoing edges of the form $(t_1\mapsto a_1, t_2\mapsto a_2)$. Once all
the nodes $t_2\mapsto a_2$ have already been visited by the algorithm, assign to
each edge $(t_1\mapsto a_1,t_2\mapsto a_2)$ a combined cost obtained by adding
to $\gamma(t_1\mapsto a_1,t_2\mapsto a_2)$ the costs of every edge reachable
from $(t_2\mapsto a_2)$, and remove all edges for which the resulting value is
not minimal. In case several edges have a minimal value assigned, choose an
arbitrary edge to remove, so that we obtain a unique edge $(t_1\mapsto a_1,
t_2\mapsto a_2)$ for $t_1\mapsto a_1$ and $t_2$.

 \item If an assignment $(t\mapsto a)$ has no incoming edges, remove it from
the cost graph.
 \end{itemize}
 The algorithm processes each edge in the polynomially sized cost graph at most
once, and thus terminates after a polynomial number of steps with an assignment
of variables to constants appearing in the compressed derivation structure.
We can use this assignment to first construct a minimal proof for $\q$ over the
compressed derivation structure, where every atom in $\q$ has its own
independent minimal proof. Note that we cannot obtain a smaller proof in
$\q$, since every atom has a minimal proof assigned, and our elimination
procedure made sure that there cannot be a different choice of assignments of
terms in $\q$ to constants that would lead to smaller proofs anywhere else.

Finally, we substitute every constant with the
corresponding Skolem term in $\Rsk(\Tmc\cup\Amc,\q)$, starting from the original
individual names and following the structure of the proof, to obtain the desired
proof of minimal tree size.
For example, we would replace an atom $P(a,b_{\exists P^-})$ that is derived in
the proof from $A(a)$ and $A\sqsubseteq\exists P$ by $P(a,f(a))$, where $f$ is
the Skolem function for the existentially quantified variable in
$A\sqsubseteq\exists P$, and replace $b_{\exists P^-}$ by $f(a)$ also in
subsequent proof steps, provided it refers to a successor of $a$. More
generally, we replace $P(t,b_{\exists P^-})$ , where $t$ is already a Skolem
term, by $P(t,f(t))$, whenever it is derived from $A(t)$ and
$A\sqsubseteq\exists P$.
\end{proof}

\TheNPHardModified*
\begin{proof}
 We first reduce SAT to $\OP\sk(\Lmc,\msize)$. Let $c_1,\ldots,c_m$ be a set of
clauses over propositional variables $p_1,\ldots,p_k$. Each clause $c_i$ is a
disjunction of literals of the form $p_j$ or $\overline{p_j}$, where the latter
denotes the negation of $p_j$. In the following, we assume clauses to be
represented as sets of literals. \Wlog we assume that for every variable $p_i$,
we also have the clause $p_i\vee\overline{p_i}$. For every variable $p_i$,
we add the facts $T(p_i)$ and $T(\overline{p}_i)$ to the ABox~\Amc. 
 For every clause $c_i$ and every literal $l_j\in c_i$, we add a fact
$c(c_i,l_j)$.
Furthermore, if $i<m$, we add the fact $r(c_i,c_{i+1})$. The
conjunctive query $\q$ is Boolean and contains for every clause $c_i$ the atoms
 $
  c(x_i^{(c)}, x_i^{(p)}), T(x_i^{(p)})
 $
 and moreover, if $i<m$, $r(x_i^{(c)}, x_{i+1}^{(c)})$,
 so that the query is tree-shaped.
We have $\Amc\models \q$ independently of whether the SAT problem has a
solution or not.

 We set our bound as $n:=2+m+(m-1)+k$, which distributes as follows in a proof 
 if the set of clauses is satisfiable. Assume that $a\colon\{p_1,\ldots,p_k\}\rightarrow\{0,1\}$ is
 a satisfying assignment for our set of clauses. 
 \begin{itemize}
  \item 2 vertices are needed for the conclusions of \gCONJ and \gEXISTS,
  \item $m$ vertices contain, for each clause $c_i$, the atom $c(c_i, l)$, where $l\in c_i$ 
  is made true by the assignment $a$,
  \item $m-1$ vertices contain the atom $r(c_i,c_{i+1})$ for each $i\in\{1,\ldots
m-1\}$,
  \item for each variable $p_i$, depending on whether $a(p_i)=1$ or $a(p_i)=0$, 
    we use either $T(p_i)$ or $T(\overline{p}_i)$. This needs another $k$ vertices. 
 \end{itemize}
 There cannot be a smaller proof since every atom in the
query needs to be matched by some ABox fact. Correspondingly, if the 
SAT problem has a solution, we can construct a proof of the desired size, and 
if there is a proof of the desired size, we can extract a solution for the SAT problem 
from it. Because the construction did not use a TBox, it follows that
$\OP\sk(\Lmc,\msize)$ is \NP-hard for any logic DL.

For domain size, we can use the bound $n:=m+k$ (covering the constants~$c_i$, $1\le i\le m$, and either~$p_i$ or~$\overline{p}_i$ for $1\le i\le k$) to achieve the same result.
%
\end{proof}


\TheELTreeUpperBounds*
\begin{proof}
  We consider again the compressed derivation structure~$\mathcal{D}$ from the proof of Lemma~\ref{lem:size-upper-dl-lite-el}.
  Recall that, for instance queries~$A(x)$, a proof of tree size $\le n$ exists in $\Rsk(\Tmc^s\cup\Amc,A(a))$ iff such a proof exists in~$\mathcal{D}$.
  By~\cite{LPAR23:Finding_Small_Proofs_for}, one can construct a proof of minimal tree size in~$\mathcal{D}$ in deterministic polynomial time.

Now consider the general case where $\q$ does not have to be an
instance query. In this case, we additionally may need to replace fresh individual names~$c_f$ in the conclusion of a proof in~$\mathcal{D}$ by appropriate Skolem terms, by checking the inferences that introduced~$c_f$.
We use this in a non-deterministic decision procedure to check whether there
exists a proof of tree size at most $k$. Specifically, we guess an assignment
of individual names to each variable in $\q$ (including both names
from $\Amc$ and the fresh ones from the compressed derivation structure). We now
use the procedure from~\cite{LPAR23:Finding_Small_Proofs_for} to compute a proof of minimal tree size for
every atom within the given tree size bound, under the guessed variable
assignment, and connect those sub-proofs into the final proof candidate. If the
resulting proof hast a tree size~$\le n$, we accept. To obtain
tractable data complexity, we adapt this procedure so that instead of guessing
the assignment, we iterate over all the possibilities, which are polynomially many in
the size of data and the TBox if the query is fixed.

For the $\NP$ lower bound, we recall from Lemma~\ref{lem:size-upper-dl-lite-el}
that proof size is bounded
exponentially, which means that the size can always
be represented using polynomially many bits. We can use this in a construction
similar as for Theorem~\ref{thm:dllite-lower} to reduce the Boolean query
entailment problem for $\EL$ to deciding the existence of a proof of bounded
size. Let $\Tmc\cup\Amc$ be an $\EL$ KB, and $\q$ Boolean query. W.l.o.g., we
assume that $\q$ contains at least one unary atom.
Let furthermore $2^n$ be a bound on the tree size of a proof for
$\Tmc\cup\Amc\models\q$ (\cf Lemma~\ref{lem:size-upper-dl-lite-el}).
We construct
$\Tmc'\cup\Amc'$ and $\q'$ s.t.\ $\Tmc'\cup\Amc'\models\q$, but only
with a proof of size less than $2^n$ if also $\Tmc\cup\Amc\models\q$.
For every variable $x$
in $\q$, we add fresh individual names $a_{x,0}$, $\ldots$, $a_{x,n}$ and $a_x$.
$\Amc'$ contains all facts of~$\Amc$,
and in addition for every atom $A(x)\in\q$ a sequence
of facts $A_0(a_{x,0})$, $s(a_{x,1},a_{x,0})$, $t(a_{x,1},a_{x,0})$,
$\ldots$, $s(a_{x,n},a_{x,n-1})$, $t(a_{x,n},a_{x,n-1})$, together with
$s_f(a_{x}, a_{x,n})$, $t_f(a_x, a_{x,n})$, with $s$, $s_f$, $t$, $t_f$ and
$A_0$ fresh. $\Tmc'$ contains all axioms of~$\Tmc$, plus additionally the
axioms $\exists s.A_0\sqcap\exists t.A_0\sqsubseteq A_0$, $\exists
s_f.A_0\sqcap\exists t_f.A_0\sqsubseteq A$.
This ensures that we can infer $A(a_{x,n})$ for every $A(x)\in\q$, but only
with a
tree proof of size $2^n$. For every binary atom $r(x,y)\in\q$, we add
$r(a_{x}, a_{y})$. We have $\Tmc'\cup\Amc'\models\q$, but if
$\Tmc\cup\Amc\not\models\q$, then every proof will be of tree size larger
than $2^n$. Since $2^n$ can be represented using $n$~bits,
our decision problem is at least as hard as Boolean query entailment, and thus
$\NP$-hard~\cite{DBLP:conf/dlog/Rosati07}.

To show \PTime-hardness for IQs, we observe that the same reduction is
possible in \LogSpace since $n$ can be represented using
logarithmically many bits, and computed using a working tape that is
logarithmically bounded (for instance by overapproximating it based on the size
of $\Tmc\cup\Amc$). Consequently, $2^n$, $\Amc'$, and $\Tmc'$ can be
computed using a Turing machine with a working tape that is bounded
logarithmically.
The result then follows from \PTime-hardness of instance query answering
in~\EL~\cite{DBLP:journals/ai/CalvaneseGLLR13}.
Regarding data complexity, we observe that $\Tmc'$ is constructed independent
of $\Amc$ and $n$, so that the same argument applies here.
%
\end{proof}


\LemSizeLower*
\begin{proof}

    For $\HornALC$, we use a similar technique as for~\EL, but this time using a binary counter enforcing a binary tree of depth $2^n$. We use concepts $A_1$, $\overline{A}_1$, $\ldots$, $A_n$, $\overline{A}_n$ as bits for the binary counter. 
    {\allowdisplaybreaks
    \begin{align*}
        \Tmc_{n,\HornALC}\ =\ \{ &&
        A &\sqsubseteq \overline{A}_1\sqcap\ldots\sqcap \overline{A}_n, \\ 
        &&\exists r.A_i &\sqsubseteq \forall r.A_i\sqcap\exists s.A_i\sqcap\forall s.A_i, \\ 
        &&\exists r.\overline{A}_i &\sqsubseteq \forall r.\overline{A}_i\sqcap\exists s.\overline{A}_i\sqcap\forall s.\overline{A}_i, \\ 
        &&\overline{A}_i \sqcap A_{i-1}\sqcap\ldots\sqcap A_0 &\sqsubseteq\exists r.A_i, \\
        &&A_i  \sqcap A_{i-1}\sqcap\ldots\sqcap A_0&\sqsubseteq\exists r.\overline{A}_i, \\
        &&\overline{A}_i \sqcap (\overline{A}_{i-1}\sqcup\ldots\sqcup\overline{A}_0 ) &\sqsubseteq \exists r.\overline{A}_i, \\
        &&A_i \sqcap (\overline{A}_{i-1}\sqcup\ldots\sqcup\overline{A}_0 ) &\sqsubseteq \exists r.A_i, \\
        &&A_1 \sqcap\ldots\sqcap A_n &\sqsubseteq B, \\ 
        &&\exists r.B \sqcap\exists s.B &\sqsubseteq B \quad\quad \mid \quad 0< i\leq n\ \}
    \end{align*}
    }
    Again, this TBox is not in normal form, but it can be normalized with only a polynomial increase in size.
    %
\end{proof}


For the next results, we make w.l.o.g. an assumption on how \gEQUAL inferences
are applied. In particular, we observe that the only equality atoms that we can
derive are due to \gMOPO inferences with a rule of shape (vi). Consequently, all
equalities we get are of the form $t=a$. In the following, we assume that in
all proofs, when \gEQUAL is applied with such an equality, we replace $t$ by
$a$ and not vice versa. This is without loss of generality as we can always
transform the proof by replacing those two terms again. Since no rule has a
Skolem function on the left, this does not change the applicability of the
rules, so that the resulting graph is still a proof and of the same size.
Moreover, we do not miss any smaller proofs in this way.

Using this assumption, we observe that binary atoms that can occur in a proof
can only be of a restricted shape: namely, $R(t,f(t))$ or $R(t,a)$, where $t$
is any ground term and $R$ can be an inverse role. This is because binary atoms
are only produced by \gMOPO inferences with a rule of shape (iv) or (vii),
as well as by \gEQUAL inferences.

\LemSizeUpper*
\begin{proof}
    \newcommand{\atoms}{\textsf{At}}

    Let $\p$ be some proof for
$\Tmc\cup\Amc\models\q$. It suffices to find a double exponential bound on the
domain size, because this also bounds the set of possible labels with unary
atoms double exponentially, and thus also the proof size.
    Let $\dom(\p)$ be the 
    Skolem terms used in $\p$, and $\atoms(\p)$ be the set of atoms used in $\p$. To every $t\in\dom(\p)$, we assign a set $\atoms(t,\p)$ of non ground atoms, where $\_$ is a special variable. 
    \begin{align*}
        \atoms(t,\p) ={} &\{A(\_)\mid A(t)\in\atoms(\p)\} \cup{} \\
        & \{R(\_,f(\_))\mid R(t,f(t))\in\atoms(\p)\} \cup{} \\
        & \{A(f(\_))\mid A(f(t))\in\atoms(\p)\} \\
        & \{R(\_,a) \mid R(t,a)\in\atoms(\p), a\in\Inds{\Tmc}\}
    \end{align*}
    Where $R$ can be a role name or an inverse role.
    New Skolem terms are only introduced via the Skolemized version of (iv),
    which produces atoms of the form $R(t,f(t))$ and $B(f(t))$. The only way to
create a role atom of a different shape is using \gEQUAL with an
equality atom
$(f(t)=t')$. Those atoms are only produced by Rule~(vi), which means that $t'$
must then be of the form $a\in\Inds{\Tmc}$.
As a result, we ensure that $\atoms(t,\p)$ refers to all atoms of the form
$r(t,t')$ and for
each such atom, all unary atoms $A(t')$ corresponding to $t'$ provided $t'$ is
not an individual name.

Assume $\dom(\p)$
contains two terms $t_1$, $t_2$ such that $\atoms(t_1,\p)=\atoms(t_2,\p)$, and
$t_1$
occurs nested within $t_2$. Inspection of the possible shapes~(i)--(vii) of
TBox rules indicates that any inference step that can be performed
using any of the atoms in $\atoms(t_2,\p)$ can also be performed using
    any of the atoms in $\atoms(t_1,\p)$. Consequently, we can simplify the
proof by
    1) removing all vertices labeled with atoms that contain the term $t_1$,
    2) adding all vertices containing the term $t_2$, as well as the edges
between them, but changing their label by
replacing $t_2$ by $t_1$, and 3) adjusting the inference steps that relied on
removed vertices, so that now they use one of the newly added vertices with the
same labels. Note that since $\atoms(t_1,\p)=\atoms(t_2,\p)$, this allows us to
keep all inference steps relying on the term $t_1$,  so that the resulting graph
is still a proof for $\Tmc\cup\Amc\models\q$. By applying this operation
repeatedly for any pair of terms $t_1$, $t_2$ such that $t_1$ occurs in $t_2$
and $\atoms(t_1,\p)=\atoms(t_2,\p)$, we ensure that no such two terms occur in
the final proof anymore. There can be at most exponentially many values for
$\atoms(t,\p)$, which means that in the final proof, the nesting depth of every
Skolem term is exponentially bounded, which bounds the overall number of Skolem
terms to double exponential \wrt the size of $\Tmc\cup\Amc$.
Since the number
of Skolem functions, as well as the elements in $\atoms(t,\p)$ are fully
determined by $\Tmc$, we furthermore obtain that the number of Skolem terms per
individual name in $\Amc$ depends only on $\Tmc$, so that the size of the
domain is polynomial in the size of the ABox.
    \qed
\end{proof}


\ThHornSROIQ*
\begin{proof}
  The claims follow straightforwardly because we only need to guess
a proof of size $n$ --- which takes non-deterministic exponential time. Since there are exponentially many facts
 over $\sig(\Tmc^s\cup\Amc)$ with a domain bounded by $n$, for $\mdomain$ we can use a similar technique.
For tree size,
it suffices to observe that the procedure described in \cite{ABB+-CADE21} to
decide the existence of proofs of bounded tree size (Theorem~17) runs in space
$p\cdot \log_2 n$, where $p$ is the maximal number of premises in any inference.
Correspondingly, this procedure would also run in \PSpace.
\end{proof}

Before we prove Theorem~\ref{th:HornSROIQ-IQ}, we show an intermediate result
for a generalization of proofs. A \emph{multi-proof} is defined like a proof,
however
\begin{enumerate}
 \item it can have more than one sink,
 \item its leafs do not have to come from the KB $\Kmc$, and
 \item it must still be connected.
\end{enumerate}

Intuitively, such a multi-proof can be a part
of a larger proof. We have the following result regarding the atoms that can
occur in a multi-proof for $\HornALCHOI$-KBs.

\begin{lemma}\label{lem:multi-proofs}
 Let $\Kmc$ be a $\HornALCHOI$ KB, \q a query and $\p$ a multi-proof in
$\Rsk(\Kmc,\q)$ such that
\begin{enumerate}
 \item $\p$ contains at least one inference,
 \item every leaf label that is an atom uses terms of nesting depth at most~$1$, and
 \item apart from the leafs, every node label that is an atom uses a term of
nesting depth $\geq 2$.
\end{enumerate}
Then, there exists a term of the form $f(a)$ such that every leaf label that is an atom uses the term $f(a)$, and every node label that is an atom
contains $f(a)$ nested within another Skolem term.
\end{lemma}
\begin{proof}
 We prove this by induction on the number of inferences in~\p.

 If $\p$ contains
exactly one inference, we observe that the premises use only terms with nesting
depth at most 1, while the conclusion must have a nesting depth $\geq 2$. It
follows that the inference must be an \gMOPO inference with a rule of type
(iv), since this is the only rule that produces a Skolem term of higher
nesting depth (note also the observation regarding \gEQUAL applications before the proof of Lemma~\ref{lem:size-upper}).
Consequently, the premises must be of the form $A(f(a))$ and
$A(x)\to R(x,g(x))\wedge B(g(x))$, and the conclusion is either
$R(f(a), g(f(a)))$ or $B(g(f(a)))$, which means that the claim is satisfied.

Assume that $\p$ satisfies the inductive hypothesis. We consider possible
extensions of $\p$. We can extend $\p$ either by
using an \gMOPO inference with a rule of type (i), (iv), (vi) or (vii)
and a sink of $\p$, or we can extend $\p$ using an \gEQUAL inference or an
\gMOPO inference with a rule of type (ii),
(iii), or (v), in which case we need to add another premise labeled with an atom,
which could be from $\p$ or from another proof that we connect to $\p$ in this
way.

In the case of using a rule of type (i), (iv), (vi) or (vii), we notice that the
conclusions either use only terms from the premise, or in case of (iv) use a term
that contains the a term from the premise nested. Consequently, it follows directly from the
inductive hypothesis for $\p$ that the extended proof satisfies the inductive
hypothesis as well.

In case of using a rule of type (ii), we notice that the other premise must use
the same term, and consequently the extended proof only satisfies the
conditions of the lemma if the other premise is derived using a proof that
already satisfies it. Specifically, we connect \p with another proof here, for
which we can assume that the inductive hypothesis holds as well. Since the
conclusion uses the same term as the premises, it follows
that the inductive hypothesis applies to the extended proof as well.

Also for rules of type (iii) and (v) we may need to extend \p using another
proof that satisfies the conditions of the lemma. Those rules are the only ones
that can also decrease the nesting depth of Skolem terms. For (iii), this is
the case if the premises are of the form $R(t, g(t))$, $A(g(t))$,
$R(x,y)\wedge A(y)\to B(x)$ and the
conclusion is of the form $B(t)$. If $t=f(a)$, then the conclusion has nesting depth
$1$, and thus the resulting proof will not satisfy the conditions of the lemma.
Otherwise, by inductive hypothesis, $t$ contains $f(a)$ nested, and the
extended proof satisfies the conditions of the lemma. For (v), the interesting
case is if the premises are of the form $A(g(t))$ and $R(g(t),t)$, and the
conclusion is $B(t)$. Here, the situation is the same as for the rule of type
(iii).

It remains to consider extensions using \gEQUAL. Here we note that, by our
assumption that we always replace complex terms by individual names, we always
obtain a term that is of the form $a=b$, $A(a)$, $r(a,t)$ or $r(t,a)$. In the
first two cases, we notice that the extended proof does not satisfy the
preconditions of the lemma anymore. In the last two cases, if it does, then $t$
must contain $f(a)$ nested by our inductive hypothesis.

We obtain that every multi-proof that satisfies the preconditions of the lemma
also satisfies its postconditions.
\end{proof}

\ThHornSROIQInstance*
\begin{proof}


Let $\q$, $n$, $\Tmc$ and $\Amc$ be the inputs to our decision problem.
We assume that $\q$ is connected since otherwise we could decide the problem
by proving the connected components independently.
We furthermore assume that the mapping of the quantified variables in the
query to the terms from the universal model is already given: for this, it
suffices to observe that, since the proof can only access Skolem terms of at
most exponential nesting depth (in the bit length of the bound $n$ on the proof
size), and the number of variables is linearly bounded by the input, there are
at most exponentially many options to try out, so that this can be easily
incorporated into an \ExpTime procedure. Since $\q$ is connected and the universal model is tree-shaped when restricted to Skolem terms of depth $\ge 1$, there exists a \emph{root term} $t_\q$ such that every variable is
matched
to a term that contains either $t_\q$ or a constant $\mathsf{a}$ nested under
at most $\lvert\q\rvert$ Skolem functions.

In the following, we fix $\_$ to be a fresh constant. Let $T_\q$ be a set of
\emph{term patterns} that are obtained from the terms matched to the variables
in the query by replacing the root term $t_\q$ by $\_$, and closing $T_\q$ under
the subterm relation. We note that $T_\q$ contains at most polynomially many
elements since the number of terms matched by the query is linearly bounded,
and each term has at most $\lvert\q\rvert$ subterms after replacing $t_\q$ by
$\_$.

%
 The idea is to construct the minimal proof from \emph{proof segments} which
can be deterministically computed using an elimination procedure. A proof
segment is a triple $\tup{t,\In,\Out,\Size}$ where $t$ is a term, $\Size$ is an
integer and $\In$ and $\Out$ are sets of atoms using only terms of the form
$t'$, $f(t')$, where $t'\in\Inds{\Tmc\cup\Amc}\cup T_\q$ and $f$ is a Skolem term 
occurring in $\Tmc_s$. The intuitive meaning of
such
a proof segment is: \enquote{It is possible to derive $\Out[\_\mapsto t]$
from
$\In[\_\mapsto t]\cup\Tmc^s$ using
at most \Size proof vertices.} To obtain an \ExpTime procedure, we add an 
additional constraint on $t$: namely, $t$ must be $\_$ or a subterm of some term 
matched to the query. Note that there are at most exponentially many terms like that.

We call a proof segment \emph{valid} if there is a
multi-proof of size \Size showing $\Out[\_\mapsto t]$ from $\In[\_\mapsto
t]\cup\Tmc^s$, and \emph{directly valid} if
there is such a multi-proof whose labels do not use
Skolem terms of nesting depth larger than $1$. Note that every directly valid
proof segment is also
valid. Direct validity of proof segments can be determined in
exponential time: for this, we note that, since
only Skolem terms of nesting
depth 1 are considered, every directly valid proof segment has a
witnessing multi-proof that has at most polynomially many nodes (one for each
possible label). We can thus enumerate all possible graphs over these nodes in
exponential time and check whether one of them corresponds to a multi-proof of
size $\Size$. Moreover, if we bound $\Size$ by the bound $n$ on the proof size
given as input, there are at most
exponentially many possible proof segments. We can thus compute in exponential
time the set of all directly valid proof segments for which
$\Size\leq n$.

We
call a proof segment $\tup{t,\In,\Out,\Size}$
\emph{initial} if $\In\subseteq\Amc$, and \emph{final} if
$\Out[\_\mapsto t]$ contains exactly the atoms of $\q$ after applying our
matching of the quantified variables.
%
For
our decision
problem, we need to determine whether there is a proof segment that is initial,
final, valid, and satisfies $\Size+2\le n$ (or $\Size+1\le n$ if \q contains only one atom), taking into account the additional inference steps with~\gCONJ and~\gEXISTS at the end of the proof.
For this, we use the following
incremental procedure
to construct new valid proof
segments from existing ones. Let $\tup{t_1,\In_1,\Out_1,\Size_1}$ and
$\tup{t_2,\In_2,\Out_2,\Size_2}$ be two proof segments. If $t_2=f(t_1)$ or $t_2=\_$, and
for the substitution
$\sigma:\_\rightarrow f(\_)$ we have
$(\In_2)\sigma\subset\Out_1$, we define
as a \emph{down-merging of $\tup{t_1,\In_1,\Out_1,\Size_1}$ and
$\tup{t_2, \In_2,\Out_2,\Size_2}$} a proof segment
\[\tup{t_1,\In_1,\Out_3,\Size_1+\Size_2-\lvert\In_2\rvert}\] where
\[\Out_3\subseteq\Out_1\cup\{\alpha\sigma\mid \alpha\in\Out_2\text{
contains only terms from }T_\q\cup\NI\}.\]
Correspondingly, if
$t_1\in\{t_2,\_\}$ and $\In_2\subseteq\Out_1$, we define as a
\emph{simple
merging of $\tup{t_1, \In_1,\Out_1,\Size_1}$ and
$\tup{t_2, \In_2,\Out_2,\Size_2}$} a proof segment
\[\tup{t_1, \In_1,\Out_3,\Size_1+\Size_2-\lvert\In_2\rvert}\] where
$\Out_3\subseteq\Out_1\cup\Out_2$. Both the result of simple merging
and of down-merging valid proof segments produce proof segments that
are also valid: for this, note that we do not need to count the nodes
corresponding to $\In_2$, as they would correspond to nodes in $\Out_1$ in the
corresponding merged proof.
Our procedure now incrementally extends our initial set of
directly valid proof segments by adding in each step all possible mergings of
proof segments in our set, where we only keep those sets for which
$\Size\leq n$. This procedure reaches a fixpoint after at most exponentially
many steps, as the set of possible proof segments is exponentially bounded. If
this fixpoint contains a proof segment that is initial and final, we accept,
and otherwise we reject.

Since our procedure only produces valid proof segments, and we only keep those
for which \Size is bounded by $n$, it is clearly sound, that is, it accepts
only if there exists a proof of size at most $n$. To show that it is also
complete, we need to prove the following claim:

\paragraph{Claim.} For every proof segment $\tup{t,\In,\Out,\Size}$, if there
exists a multi-proof $\p$ of size $\Size$ with premises
$\In\cup\Tmc_s$ and conclusions $\Out$, then
$\tup{t,\In,\Out,\Size}$ can be obtained through a sequence of merging
operations starting from directly valid proof segments.

\paragraph{Proof of claim.} Let $\p$ be as in the claim. We prove the claim by
induction over the maximal nesting depth of any term occurring in $\p$.
If the nesting depth is $1$, the claim directly holds since
$\tup{t,\In,\Out,\Size}$ is already directly valid. Otherwise, assume that the
claim holds for all proofs for which the maximal nesting depth of any term
occurring is $k$.

We decompose $\p$
into different, possibly overlapping, multi-proofs $\p'$ that satisfy the
following conditions:
\begin{enumerate}
 \item $\p'$ is connected, 
 \item either (a) all labels use only terms of nesting depth $\leq 1$, or (b) only the
premises and the conclusions use terms of nesting depth $\leq 1$, and at least one label uses a higher nesting depth, 
and
 \item there is no proper subproof of $\p'$ that satisfies the above
conditions.
\end{enumerate}

If $\p'$ satisfies~2(a), then it corresponds to a directly valid proof segment.
If it satisfies~2(b), by Lemma~\ref{lem:multi-proofs},
we obtain the following two properties:
\begin{itemize}
  \item All leafs of $\p'$ with atom labels use a term of the form $f(t)$.
  \item Every internal vertex in $\p'$ with atom label must use a term that
  contains $f(t)$.
\end{itemize}
%
%
%
%
We
obtain that $\p'$ can be represented by a proof segment
$\tup{\_,\In',\Out',\Size'}$, where $\Size'$ is the size of $\p'$, and $\In'$
and $\Out'$ are obtained from the premises and conclusions in $\p'$ by
replacing $f(t)$ by $\_$. Since by this replacement, we would also reduce the
nesting depth of $\p'$, we obtain that $\tup{\_,\In',\Out',\Size'}$ satisfies the
claim by our inductive hypothesis.

We obtain that all multi-proofs $\p'$ in
our decomposition have a corresponding proof segment that satisfies our claim.
Observing that the only possible overlaps between these multi-proofs are by the
conclusions (which must be contained in the premises of another multi-proof 
or contribute to the final conclusion),
we observe that the proof segment $\tup{t,\In,\Out,\Size}$ can be produced from
these proof segments by applying down-merging and simple merging operations,
so that finally, $\tup{t,\In,\Out,\Size}$ satisfies the claim as well.
\hfill $\blacksquare$

We obtain that our method is sound and complete, and runs in exponential time,
which completes the proof.
%
%
%
%
\end{proof}

\subsection{\nameref{sec:cq}}

\ThmTransformation*
\begin{proof}
Assume that we have a proof $\p$ in $\Rcq(\Tmc\cup\Amc,\q)$.
We first defer all applications of schema~\EXISTS to the very end of the proof, which is possible since all other inferences remain applicable to any instance $\exists\vec{x}.\,\phi(\vec{x},\vec{a})$ of a CQ $\exists\vec{x},\vec{y}.\,\phi(\vec{x},\vec{y})$. This transformation can only decrease the size of the proof.
We then recursively change the labeling function so that it uses
Skolem terms rather than quantified variables. Specifically, starting from the leafs of~\p, we adapt
inferences of schema~\MOPO so that instead of a rule
$\psi(\vec{y},\vec{z})\to\exists\vec{u}.\,\chi(\vec{z},\vec{u})\in\Tmc$, the
corresponding Skolemized version
$\psi(\vec{y},\vec{z})\to\chi'(\vec{z})\in\Tmc^s$ is used, with existentially
quantified variables in the conclusion replaced by the corresponding Skolem
terms. The resulting hypergraph has the same size, since we only changed the
labeling, and moreover all CQs are ground.
In this process, whenever we apply a rule $\phi(\vec{x},\vec{y})\to\exists\vec{x}.\,\phi(\vec{x},\vec{y})$ generated by \TAUT to a ground CQ using \MOPO, we can omit this subproof since $\exists\vec{x}.\,\phi(\vec{x},\vec{y})$ is satisfied by the same ground atoms used to match $\phi(\vec{x},\vec{y})$.
Next, we split each modified \MOPO inference into a corresponding set of \gMOPO inferences -- this replaces each vertex~$v$ by at most $|\Tmc|$ vertices (at most for each new atom derived from the right-hand size of a rule in~\Tmc),
and thus increases the size of the hypergraph by a factor polynomial in the size of~$\p$.
Even though \Rsk contains no version of the \TAUT inference schema to generate copies of atoms, we can always use the same ground atom several times in the same inference if necessary.
At the same time, we remove all \CONJ inferences (since \gMOPO anyway uses multiple 
premises instead of a conjunction) and replace them by a single application of \gCONJ 
(conjoining all atoms relevant for~$\q(\vec{a})$).
This can also only decrease the size of the proof.
Similarly, the final \EXISTS step becomes an instance of \gEXISTS that produces the final conclusion $\q$.
We obtain a proof in~$\Rsk(\Tmc^s\cup\Amc,\q)$ in time polynomial in the size of~$\p$.

Now assume that we have a proof $\p$ in $\Rsk(\Tmc^s\cup\Amc,\q)$.
We transform $\p$ into a semi-linear proof in $\Rsk(\Tmc\cup\Amc,\q)$.
In the beginning, we keep the Skolemization, which is eliminated in the last step.
We first collect any
ABox facts from $\Amc$ that are used in $\p$ into a single CQ using \CONJ.
Then we aggregate inferences of \gMOPO into \MOPO-inferences. Specifically,
provided that for an \gMOPO inference $(S,d)$, all labels of nodes in $S$
occur on a node $v$ that we have already aggregated, we collect all inferences
of the form $(S,d')$ (same premises, different conclusion), and transform them
into a single inference using (a Skolemized version of) \MOPO.
When $S$ contains the same atom multiple times, we first generate an appropriate number of copies using \TAUT and \MOPO; the number of such additional proof steps for each \MOPO-inference is bounded by~$|\Tmc|$ (more precisely, the maximum number of atoms on the left-hand side of any rule in~\Tmc).
%
Depending on which premises are still
needed in later inferences, we keep them in the conclusion of each \MOPO-inference or not. Since we
keep all atoms together in each step, the final application of \gCONJ is not
needed anymore. The resulting proof is de-Skolemized by replacing Skolem terms
by existentially quantified variables. This transformation
is again polynomial in the size of the original proof: the number of initial
applications of \CONJ is bounded by the size of $\p$, the aggregated \MOPO-steps can only decrease the size of the proof, and for each of these steps, we need at most $|\Tmc|$ additional \TAUT- and \MOPO-steps.
Note also that this process always yields a tree-shaped proof.
\end{proof}

Due to Theorem~\ref{thm:transformation}, many of the results we have seen before also apply to \Rcq (except for domain size, which is not defined in this context).
For example, the arguments in the proof of Theorem~\ref{th:TheACzero} apply in the same way here, \ie one only has to consider constantly many ABoxes (modulo isomorphisms) to search for proofs below a given size bound, and similarly for Theorem~\ref{th:TheACzeroRewritable}.
Lemma~\ref{lem:size-upper-dl-lite-el} also holds for~\Rcq because of Theorem~\ref{thm:transformation}, which gives us the same upper bounds as in Theorem~\ref{thm:dllite-el-upper-bounds}.
Moreover, Theorem~\ref{thm:transformation} also ensures that the arguments for \NP-hardness in Theorem~\ref{thm:dllite-lower} can be transferred to~\Rcq.
We can similarly transfer the results of Lemma~\ref{lem:size-upper} and Theorem~\ref{th:HornSROIQ}; however, the current proof of Theorem~\ref{th:HornSROIQ-IQ} works only for \Rsk.


Due to duplication of atoms via~\TAUT, Theorem~\ref{the:np-hard-trees} can also be shown for \Rcq and \mtree (and so Theorem~\ref{th:ptime-combined} cannot hold for \Rcq):
\ThmNPHardCQStructure*
\begin{proof}
  We follow a similar approach as in Theorem~\ref{the:np-hard-trees}.
 The central observation is that we can simulate in \Rcq the behavior of
\gCONJ by copying atoms. Specifically, fix an inference of \gCONJ
with premises $\alpha_1(\vec{t}_1)$, \dots, $\alpha_n(\vec{t}_n)$.
In~\Rsk these atoms would not occur separately, but together as a non-ground CQ $\exists\vec{x}.\alpha_1(\vec{x}_1)\land\dots\land\alpha_n(\vec{x}_n)$.
Now assume that $\alpha_1(\vec{t}_1)=\alpha_2(\vec{t}_2)$, in which case \Rsk can use the same vertex as both the first and the second premise of the inference.
In \Rcq, this corresponds to only $\alpha_1(\vec{x}_1)$ occurring in the current CQ, although both $\alpha_1(\vec{x}_1)$ and $\alpha_2(\vec{x}_2)$ (using the same predicate, but different variables) are needed for a subsequent inference.
To obtain a similar effect as in $\Rsk$, we can first use \TAUT to derive
$\alpha_1(\vec{x}_1)\rightarrow\exists\vec{x}_1.\,\alpha_1(\vec{x}_1)$, which is then used
with \MOPO and the current CQ~$\phi(\vec{x}$) to obtain a query in which both
$\alpha_1(\vec{x}_1)$ and $\alpha_2(\vec{x}_2)$ occur (recall that \MOPO can be used
in such a way that the new atoms are not replacing others, but are simply
added). This way, we can duplicate an arbitrary number of atoms and variables
using two steps.

Let $\Amc$ and $\q$ be as in the proof of Theorem~\ref{the:np-hard-trees}, and assume 
the SAT problem has a satisfying assignment~$a$.
We can then
construct a tree proof as follows, where we collect the different atoms, one after the other, using \CONJ into a single query. Since these inferences always have two 
premises, we only count the leaves in the following, as a binary tree with $\ell$ leaves and all 
other vertices binary always has $2\ell-1$ vertices in total. Specifically, this means that the 
number of vertices is independent of how we organize the inferences listed in the following.  
\begin{itemize}
  \item We need to collect all clauses of the form $r(c_i,c_{i+1})$, and $c(c_i,l_j)$, where 
   $l_i$ is some literal in $c_i$ evaluated to true under the chosen assignment $a$. This gives 
   $2m-1$ leaves.
  \item Another $k$ leaves are needed to add, for each variable $p_i$, $T(p_i)$ if $a(p_i)=1$,
  and $T(\overline{p}_i)$ if $a(p_i)=0$.
%
\end{itemize}
Since this makes $2m-1+k$ leaves, we obtain that the corresponding proof must have 
$4m+2k-3$ vertices in total, independently of how we organize these inferences. We instantiate 
\TAUT with $\q(\vec{x})\rightarrow\exists\vec{x}.\,\q(\vec{x})$, where 
$\q(\vec{x})$ is our query, containing as quantified variables exactly $\vec{x}$.
Since the left-hand side matches our ground query constructed so far, we apply \MOPO 
as final step to produce the conclusion, obtaining a tree proof of size $n:=4m+2k-1$.

As in the proof of Theorem~\ref{the:np-hard-trees}, we argue that there cannot possibly be a smaller proof for the query, since 
every leaf of this proof has to be used.
Moreover, \EXISTS cannot be used to construct the final CQ since it cannot duplicate atoms and at least one atom of the form $T(p_i)$ or $T(\overline{p}_i)$ has to be used twice due to the additional clauses $p_i\lor\overline{p}_i$.
Consequently, if we find a proof of tree size~$n$, we 
can construct a satisfying assignment from it, and if there is a satisfying assignment, we 
can construct a proof of size~$n$. It follows that $\OP\cq(\emptyset,\mtree)$ must be \NP-hard.
The same arguments apply to \msize since the proofs constructed above are tree-shaped and all vertices have different labels.
\end{proof} 
\end{document}